\documentclass{article}
\usepackage[utf8]{inputenc}

\usepackage{hyperref}
\usepackage{cite}

\usepackage{geometry} 

\usepackage{amssymb} 
\usepackage{amsmath, bm} 
\usepackage{amsthm} 
\usepackage{mathtools}
\usepackage{scrextend} 
\usepackage[framemethod=tikz]{mdframed}  
\usepackage{enumitem} 
\usepackage{slashed} 
\usepackage{verbatim} 
\setlist[enumerate]{itemsep=0.5pt, topsep=4pt}

\usepackage{multicol} 
\usepackage{cancel} 

\usepackage{float}

\newcommand{\lv}{\vspace{0.3cm}}
\newcommand{\vv}{\vspace{0.5cm}}
\newcommand{\vvv}{\vspace{1cm}}

\newcommand{\ddt}{\left. \frac{d}{dt} \right|_{t=0}}

\newcommand{\R}{\mathbb{R}}

\newcommand{\Z}{\mathbb{Z}}

\newcommand{\dd}[1]{\frac{\partial}{\partial #1}}
\newcommand{\norm}[1]{\left\lVert#1\right\rVert}

\newcommand{\normmC}[4]{ \left\lVert#1\right\rVert_{C, (#2 \to  #3),#4}}
\newcommand{\normmH}[4]{ \left\lVert#1\right\rVert_{H, (#2 \to  #3),#4}}
\newcommand{\normmmC}[3]{ \left\lVert#1\right\rVert_{C, #2,#3}}
\newcommand{\normmmH}[3]{ \left\lVert#1\right\rVert_{H, #2,#3}}

\newcommand{\gsc}{\mathfrak{g_{sc}}}
\newcommand{\scc}{\mathfrak{sc}}
\newcommand{\trksc}{trK_{\mathfrak{g_{sc}}}}
\newcommand{\trko}{trK_{\mathfrak{B}}}
\newcommand{\gammasc}{\gamma_{\mathfrak{g_{sc}}}}
\newcommand{\gammao}{\gamma_{\mathfrak{B}}}
\newcommand{\trkpio}{ \mathrm{tr}_{\partial M} \Pi_{\mathfrak{B}}}
\newcommand{\omegao}{\omega_{\mathfrak{B}}}

\newcommand{\gsptime}{\mathfrak{g}^{(4)}}
\newcommand{\gsp}{\mathfrak{g}}
\newcommand{\ngf}{\mathbf{n}_{\mathfrak{g}}}
\newcommand{\ngg}{\mathbf{n}_{g}}

\newcommand{\bartnik}{\gamma_{\mathfrak{B}}, \frac{1}{2} trK_{\mathfrak{B}}, \mathrm{tr}_{\partial M} \Pi_{\mathfrak{B}}, \omega_{\mathfrak{B}}} 
\newcommand{\bartnikf}{{\gamma_{\mathfrak{B}}}_{f}, \frac{1}{2} {trK_{\mathfrak{B}}}_f, {\mathrm{tr}_{\partial M} \Pi_{\mathfrak{B}}}_f, {\omega_{\mathfrak{B}}}_f} 

\newcommand{\HtoH}[3]{H^{#1}_{#2}\left([r_0,\infty); H^{#3}(S^2)\right)}

\newcommand{\LtoH}[2]{L^{2}_{#1}\left([r_0,\infty); H^{#2}(S^2)\right)}
\newcommand{\CtoH}[3]{C^{#1}_{#2}\left([r_0,\infty); H^{#3}(S^2)\right)}

\newcommand{\LtoHr}[2]{L^{2}_{#1}\left([r_0,\infty); H^{#2}(S^2)\right)}
\newcommand{\CtoHr}[3]{C^{#1}_{#2}\left([r_0,\infty); H^{#3}(S^2)\right)}

\newcommand{\HtoHH}[3]{H^{#1}_{#2}\left([r_0,\infty); \mathcal{H}^{#3}(S^2)\right)}

\newcommand{\AH}[3]{{\mathcal{A}_{H}}^{(#1,#2)}_{#3}}
\newcommand{\AC}[3]{{\mathcal{A}_{C}}^{(#1,#2)}_{#3}}
\newcommand{\AHC}[3]{\mathcal{A}^{(#1,#2)}_{#3}}

\newcommand{\sgsc}{\star_{g_{sc}}}
\newcommand{\dy}{\slashed dY_{m\ell}}
\newcommand{\sdy}{\slashed \star \slashed d Y_{m\ell}}
\newcommand{\dsig}{d\sigma_{\mathbb{S}^2}}
\newcommand{\summ}{\sum_{\ell=0}^{\infty} \sum_{m=-\ell}^{\ell}}
\newcommand{\summm}{\sum_{\ell=1}^{\infty} \sum_{m=-\ell}^{\ell}}

\numberwithin{equation}{section}
\newtheorem{thm}{Theorem}[section]
\newtheorem{lem}[thm]{Lemma}
\newtheorem{prop}[thm]{Proposition} 
\newtheorem*{prop*}{Proposition}
\newtheorem {cor}[thm]{Corollary}  
\newtheorem{conj}[thm]{Conjecture}
\newtheorem*{question*}{Question}

\newtheorem*{mainthm}{Main Theorem}{\bf}{\it}
\newtheorem*{redthm}{Reduction Theorem}{\bf}{\it}

\theoremstyle{definition}
\newtheorem{defn}[thm]{Definition}

\theoremstyle{remark}
 
\newtheorem{remark}[thm]{Remark}

\title{ Local Well-Posedness for the Bartnik Stationary Extension Problem near Schwarzschild Spheres}

\author{Ahmed Ellithy\thanks{Department of Mathematics, Uppsala University. Email: \texttt{ahmed.ellithy@math.uu.se}}}
\date{}

\begin{document}

\maketitle
\begin{abstract}

We investigate the Bartnik stationary extension conjecture, which arises from the definition of the spacetime Bartnik mass for a compact region in a general initial data set satisfying the dominant energy condition. This conjecture posits the existence and uniqueness (up to isometry) of an asymptotically flat stationary vacuum spacetime containing an initial data set $(M, \gsp, \Pi)$ that realizes prescribed Bartnik boundary data on $\partial M$, consisting of the induced metric, mean curvature, and appropriate components of the spacetime extrinsic curvature $\Pi$.

Building on the analytic framework developed in \cite{ahmed} for the static case, we show that, in a double geodesic gauge, the stationary vacuum Einstein equations reduce to a coupled system comprising elliptic and transport-type equations, with the genuinely stationary contributions encoded in an additional boundary value problem for a $1$-form $\theta$. 

We establish local well-posedness for the Bartnik stationary metric extension problem for Bartnik data sufficiently close to that of any coordinate sphere in any initial data set (possibly non time-symmetric) in Schwarzschild spacetime. This includes spheres arbitrary close to the apparent horizon in the initial data set. A key feature of our framework is that the linearized equations decouple: the equations for the metric and potential reduce to the previously solved static case, while the boundary value problem for $\theta$ is treated independently. We prove solvability of this boundary value problem in the Bochner-measurable function spaces adapted to the coupled system developed in \cite{ahmed}, establishing uniform estimates for the vector spherical harmonic decomposition of $\theta$.
\end{abstract}
\tableofcontents

\section{Introduction}

The concept of quasi-local mass in general relativity has been a central problem in mathematical relativity for several decades (see \cite{penrose-problem}), where one aims to capture the total mass contained within a bounded region of a time-slice of a given spacetime in a way that is both physically meaningful and mathematically rigorous. Among the various proposals is the Bartnik mass, originally introduced by Bartnik in \cite{Bartnik2} in the case of time-symmetric initial data sets (see also \cite{mccormick, anderson2023} for a survey on this topic). In this case, the Bartnik mass is defined for a compact Riemannian 3-manifold $(\Omega, g)$ with boundary and nonnegative scalar curvature. The Bartnik mass of $(\Omega, \gsp)$ is then the infimum of the ADM masses among all admissible asymptotically flat extensions $(M_{\mathrm{ext}}, \gsp_{\mathrm{ext}})$ of $(\Omega,\gsp)$ such that the glued Riemannian manifold $M = M_{\mathrm{ext}} \cup \Omega$ satisfies suitable geometric conditions (the metrics match so that the glued manifold has nonnegative scalar curvature and contains no horizons enclosing $\Omega$). Importantly, Bartnik demonstrated in \cite{Bartnik3} that this infimum depends only on the geometry of the boundary $\partial \Omega$ --- namely, the Bartnik data $(\gamma, H)$, where $\gamma := \gsp|_{\partial \Omega}$ and $H$ is the mean curvature.

A fundamental conjecture due to Bartnik in \cite{Bartnik3}, now called the Bartnik static minimization conjecture, asserts that this infimum is attained by a (geometrically) unique asymptotically flat static vacuum extension --- that is, a time-symmetric initial data set embedded as a hypersurface in a static vacuum spacetime. This conjecture leads naturally to the Bartnik static metric extension conjecture: Given a closed Riemannian surface $(\Sigma, \gamma)$ and a function $H$ on $\Sigma$, there exists a unique asymptotically flat Riemannian manifold $(M,\gsp)$ with boundary $\partial M = \Sigma$ and a positive function $f$ (the static potential), such that $g|_{\partial M} = \gamma$ and the mean curvature is $H$ at the boundary, and such that the pair $(g,f)$ solves the static vacuum Einstein equations:
\begin{equation}
\mathrm{Ric}_{\gsp} = f^{-1} \mathrm{Hess}_{\gsp} f, \qquad \Delta_{\gsp} f = 0
\end{equation}
These equations are equivalent to the statement that $(M,g)$ arises as a time-symmetric initial data set in the static spacetime $(\mathbb{R} \times M, \mathbf{g}^{(4)} = -f^2 dt^2 + \gsp)$, which satisfies the Einstein vacuum equations $\mathrm{Ric}_{\gsptime} = 0$. Thus, the conjecture reduces the Bartnik minimization problem to a boundary value problem for the static vacuum equations with prescribed Bartnik data. The extension conjecture then serves as a test for the Bartnik minimization conjecture.

There has been substantial progress on these conjectures in recent years. Corvino, and independently Anderson–Jauregui (\cite{corvino1}, \cite{anderson-jauregui}), proved that any minimizer of the Bartnik mass must indeed be a static vacuum extension, should it exist. Anderson and Khuri (\cite{anderson-khuri}) showed, via counterexamples, that the most optimistic version of global well-posedness for the boundary value problem fails. Much research has since focused on the local well-posedness of the static extension problem near model data --- particularly, near coordinate spheres in Euclidean space or Schwarzschild manifolds (see, e.g., \cite{anderson-local-exist,AhmedAn,an-huang,an-huang2,an-huang3}). In recent work by the present author (\cite{ahmed}), a new framework was established for analyzing the local extension problem near arbitrary Schwarzschild spheres, including those with mean curvature arbitrarily close to zero. In this approach, the vanishing of the Weyl tensor in three dimensions is exploited to reduce the problem, in a geodesic gauge, to a coupled elliptic–transport system in suitable function spaces. A key innovation was the use of certain Bochner-measurable function spaces, traditionally employed in the study of evolution equations, to handle the coupled nature of the problem (see \cite{ahmed} for a discussion of the differences between this framework and the approaches developed in \cite{an-huang,an-huang2,anderson-khuri}).  More specifically, the spaces we used for the function $u:= \ln f$ are  $\AHC{2}{k}{\delta}(M) := \AC{2}{k}{\delta}(M) \cap \AH{2}{k}{\delta}(M)$ defined by (see definition \ref{spaces for u})

\begin{equation*} u \in \AH{2}{k}{\delta}(M) \quad \Longleftrightarrow \quad \begin{cases} \begin{aligned}  &u \in \LtoHr{\delta}{k}\\&\partial_r u \in \LtoHr{\delta-1}{k-1} \\&\partial_r^2 u \in \LtoHr{\delta-2}{k-2} \end{aligned} \end{cases} \end{equation*}
\begin{equation*} u \in \AC{2}{k}{\delta}(M) \quad \Longleftrightarrow \quad \begin{cases} \begin{aligned}  &u \in \CtoHr{0}{\delta}{k}\\&\partial_r u \in \CtoHr{0}{\delta-1}{k-1} \\&\partial_r^2 u \in \CtoHr{0}{\delta-2}{k-2} \end{aligned} \end{cases} \end{equation*}
where $r_0>0$, $k\geq 2$, and $\delta \in (-1,-\frac{1}{2})$ is a weight introduced appropriately in the norms of the above spaces to control the decay at infinity. In chapter 3 in  \cite{ahmed}, we established the solvability of a certain elliptic problem in the above spaces, which was needed to prove the solvability of the linearized problem. More specifically, defining the operator $Q: u \mapsto (\Delta_g u, u|_{\partial M})$ with respect to a certain asymptotically flat metric $g$ on $M = \R^3 \setminus B_{r_0}$, we demonstrated that

\[Q:\AH{2}{k}{\delta}(M) \to \LtoHr{\delta-2}{k-2}  \times H^{k-1/2}(\partial M) \qquad \text{is an isomorphism}\]
\[Q:\AC{2}{k}{\delta}(M) \to \CtoHr{0}{\delta-2}{k-2}  \times H^{k}(\partial M) \qquad \text{is an isomorphism} \]

For the reader's convenience, we recall the main result of \cite{ahmed} (see the reference for the definition of the spaces and a more precise formulation): 

\begin{thm}
Let $M:= \R^3\setminus B_{r_0}$ where $r_0>2m_0$ and $m_0>0$. Let $\delta \in (-1,-\frac{1}{2}]$ and $k\geq 5$. Let $(M, \gsc)$ be the Riemannian Schwarzschild manifold of mass $m_0$. There exists a neighbourhood $\mathcal{U}$ of the Bartnik data $( \gammasc, \frac{1}{2} \trksc)$ on $\partial M$ of $(M, \gsc)$ in $$ \mathcal{M}^{k+1}(S^2) \times H^k(S^2) $$ such that for every $(\gamma, H) \in \mathcal{U}$, there exists a vacuum static spacetime $(\R\times M, \gsptime)$, unique up to isometry, close to the Schwarzschild spacetime metric on $\R \times M$ in a certain Banach space that contains a time-symmetric (totally geodesic) spacelike hypersurface with Bartnik data $(\gamma, H$) on $\partial M$. More precisely, there exists a unique metric $g$ and a function $u$ on $M$ such that 
\begin{itemize}
\item $g$ can be written globally in the form $g = dr^2+g(r)$, where $r = \mathrm{dist}(\cdot, \partial M)+r_0$ and $g(r)$ is the induced metric on the level sets of $r$. 
\item The spacetime metric $\gsptime := -e^{2u} dt^2 + e^{-2u} g$ on $\mathcal{M}:= \R \times M$ satisfies Einstein's vacuum equations, i.e. $\mathrm{Ric}_{\gsptime} = 0$. 
\item The $\{t=0\}$ hypersurface  $(M, e^{-2u} g)$  in $(\mathcal{M}, \gsptime)$ satisfies the desired Bartnik boundary conditions, i.e. the Bartnik data of $\partial M$ in $(M, e^{-2u} g)$ is $(\gamma, H)$. 
\end{itemize}

\end{thm}

\vv

\subsection*{The Spacetime Bartnik Mass}

Bartnik later in \cite{Bartnik3} extended his definition of quasi-local mass to the more general case outside time-symmetry, where one now considers compact regions in a general initial data set $(M,\gsp,\Pi)$ satisfying the dominant energy condition (the conditions of the positive mass theorem).

Note that in the time-symmetric case (when $\Pi = 0$), this reduces to the condition that $(M,\gsp)$ has nonnegative scalar curvature. In this broader setting, we consider a region $(\Omega, \gsp, \Pi)$ in an initial data set $(M,\gsp,\Pi)$ satisfying the above equations. The Bartnik mass of $(\Omega,\gsp,\Pi)$ is defined as the infimum of the ADM masses over all initial data sets $(M_{\mathrm{ext}}, \gsp_{\mathrm{ext}}, \Pi_{\mathrm{ext}})$ extending $(\Omega, \gsp, \Pi)$, such that the glued initial data set satisfies the conditions of the positive mass theorem, and contains no apparent horizons enclosing $\Omega$. Analogously to the time-symmetric case, Bartnik showed that this infimum depends only on certain geometric properties at the boundary $\partial \Omega$—the so-called Bartnik data $( \gamma, H, \mathrm{tr}_{\partial \Omega} \Pi, \omega_{\partial \Omega})$, where $\gamma := g|_{\partial \Omega}$, $H$ is the mean curvature, $\mathrm{tr}_{\partial \Omega} \Pi$ is the boundary trace of $\Pi$, and $\omega_{\partial \Omega}$ is the restriction of $\Pi$ to the normal bundle of $\Omega$ at the boundary, defined by the 1-form
\[
\omega_{\partial \Omega} (v) := \Pi(\nu, v),
\]
where $\nu$ is the unit normal vector field on $\partial \Omega$ and $v$ is any vector tangent to $\partial \Omega$.

The natural analog of the Bartnik static minimization conjecture in this more general context is the Bartnik stationary minimization conjecture, which asserts that the infimum is attained by a (geometrically) unique stationary vacuum extension: an asymptotically flat initial data set $(M, \gsp, \Pi)$ embedded as a spacelike hypersurface in a globally hyperbolic stationary vacuum spacetime $(\mathcal{M}, \gsptime)$. A spacetime is called stationary if it admits a timelike Killing vector field. In the so-called quotient formalism (see \cite{anderson-stationary,beig-simon}), the stationary vacuum spacetime metric can be written as
\begin{equation}
\mathbf{g}^{(4)} = -f^2 (dt+\theta)^2 + f^{-2}g,
\end{equation}
where $\dd{t}$ is the Killing vector field, $f > 0$ is a function, $\theta$ is a 1-form, and $g$ is a Riemannian metric on $M$, all independent of $t$. We can assume without loss of generality that the stationary vacuum extension $(M, \gsp,\Pi)$ is simply the $\{t=0\}$ hypersurface in $(\mathcal{M}, \gsptime)$. In particular, $\gsp$ and $\Pi$ can be written in terms of $g$, $f$ and $\theta$ as follows (see \cite{chrusciel}): 

\[ \gsp = f^{-2} g - f^2 \theta \otimes \theta, \qquad \Pi = \frac{N}{2}  \mathcal{L}_{\bar \theta^{\sharp_{\mathfrak{g}}}} \mathfrak{g}\]

where 

\[N = \frac{\sqrt{1-f^4|\theta|_g^2}}{f}, \quad \bar \theta = f^2 \theta\]

The stationary vacuum equations $\mathrm{Ric}_{\mathbf{g}^{(4)}} = 0$ are equivalent to the following system on $M$ (see \cite{anderson-stationary,beig-simon}):
\begin{equation} \label{ric=0}
\begin{aligned}
&\mathrm{Ric}_g = \frac{1}{f}\mathrm{Hess}_g (f) + 2 f^{-4} ( \eta \otimes \eta - |\eta|_g^2 \cdot g), \\
&\Delta_g f = -2 f^{-3} |\eta|_g^2, \\
&d\eta = 0,
\end{aligned}
\end{equation}
where $\eta = -\frac{1}{2} f^4 \star_g d\theta$ is the twist 1-form measuring the failure of the spacetime to be static.

In analogy with the static case, the stationary extension conjecture leads to the following boundary value problem:

\begin{question*}[Bartnik stationary metric extension problem]
\textit{Given Bartnik boundary data $(\gamma, H, \mathrm{tr}_\Sigma \Pi, \omega_\Sigma)$ on a 2-sphere $\Sigma$, does there exist a unique (up to isometry) asymptotically flat stationary vacuum spacetime $(\mathcal{M}, \gsptime)$ containing an initial data set $(M, \gsp, \Pi)$, such that the induced Bartnik data on $\partial M = \Sigma$ matches the prescribed data?}
\end{question*}

Three important examples of stationary vacuum spacetimes are as follows:
\begin{itemize}
\item The Minkowski spacetime $(\mathbb{R}^4, \eta)$, with $\eta = -dt^2 + dx_1^2 + dx_2^2 + dx_3^2$, is the simplest stationary (in fact, static) vacuum solution.
\item The Schwarzschild spacetime, with metric
\[
\gsptime_{sc} = -\left(1 - \frac{2m}{r}\right)dt^2 + \left(1 - \frac{2m}{r}\right)^{-1}dr^2 + r^2(d\theta^2 + \sin^2\theta\, d\phi^2),
\]
is a static vacuum solution and provides the main model for the Bartnik extension problem near Schwarzschild spheres.
\item The Kerr spacetime, with metric
\[
\gsptime_{Kerr} = -\left(1 - \frac{2mr}{\Sigma}\right)dt^2 - \frac{4mar \sin^2\theta}{\Sigma} dt d\phi + \frac{\Sigma}{\Delta} dr^2 + \Sigma d\theta^2 + \left(r^2 + a^2 + \frac{2ma^2 r \sin^2\theta}{\Sigma}\right)\sin^2\theta d\phi^2,
\]
where $\Sigma = r^2 + a^2\cos^2\theta$ and $\Delta = r^2 - 2mr + a^2$, is a non-static (truly stationary) vacuum solution.
\end{itemize}

The stationary extension problem is significantly less understood than its static counterpart. An, in \cite{an-elliptic,an-elliptic2}, established the ellipticity of the stationary vacuum equations with natural boundary data in harmonic gauge. Recent work by Huang and Lee ( see \cite{huang-lee}) addresses conditions under which the Bartnik mass infimum is attained by a stationary extension. The stationary problem is substantially more intricate, both analytically (due to the nontrivial coupling between the metric, the potential, and the twist 1-form) and geometrically (due to the gauge freedom in choosing $\theta$ and the increased complexity of the boundary data).

\subsection*{Summary of the Present Work}

In this paper, we extend the analysis of the Bartnik extension problem beyond time symmetry, establishing the first local well-posedness result for the Bartnik stationary metric extension problem near \emph{every} coordinate Schwarzschild sphere in a time-symmetric slice of Schwarzschild spacetime. Moreover, by identifying how the Bartnik data transform under Lorentz boosts of the spacetime normal bundle of $\partial M$, we deduce an immediate corollary giving the same local well-posedness for Bartnik data prescribed on any spacelike graph hypersurface $\{t=f(x)\}$ in Schwarzschild with $f|_{\partial M} = 0$. Our approach builds on the framework developed in \cite{ahmed} for the static case, adapting it to handle the additional complexities introduced by the twist 1-form that characterizes genuinely stationary solutions.

The stationary vacuum Einstein equations are significantly more intricate than their static counterpart due to the coupling between the metric $g$, the function $u$, and the 1-form $\theta$ appearing in the spacetime metric: 
\[\gsptime = -e^{2u} (dt+\theta)^2 + e^{-2u} g\]

A crucial insight is that, despite this coupling, the linearized equations in a certain gauge choice exhibit a decoupling property that we exploit in our analysis. 

\subsubsection*{Choice of Gauges}

The stationary extension problem has inherent gauge freedom that must be fixed to ensure well-posedness. Given Bartnik data on the boundary, there are multiple ways to represent a stationary extension:

\begin{itemize}
\item \textbf{Spatial gauge freedom}: The choice of coordinates on $M$, which affects the form of the metric $g$.
\item \textbf{Time-slice gauge freedom}: The choice of initial data set within the stationary spacetime, which corresponds to adding an exact 1-form $df$ to $\theta$ (see Lemmas \ref{f-prop} and \ref{f=0-lem}).
\end{itemize}

\noindent We fix these gauge freedoms as follows:

\vv

\textbf{The $g$-geodesic gauge:} Following the static case addressed in \cite{ahmed},  the metric $g$ is globally written in geodesic coordinates:
\[g = dr^2 + g(r) \]
where $r = \mathrm{dist}(\cdot, \partial M) + r_0$ and $g(r)$ is a one-parameter family of metrics on $S^2$. This gauge together with the vanishing of the Weyl tensor in 3 dimensions reduces the vacuum equations to a coupled elliptic-transport system, where the evolution of the geometry is governed by ODEs in the radial direction. For the Schwarzschild solution, the metric $g$, denoted by $g_{sc}$ becomes 
\[g_{sc} = f_{sc}^2 \gsc = dr^2 + r(r-2m_0) \gamma_{\mathbb{S}^2}\]

\textbf{The $\theta$-geodesic gauge:} The gauge for $\theta$ is a key new ingredient in the stationary setting. Unlike the metric, the 1-form $\theta$ retains a residual gauge freedom corresponding to redefining the time coordinate by $t \mapsto t+f(x)$ for functions $f$ vanishing (together with $df$) on $\partial M$ to preserve the Bartnik data. To uniquely fix $\theta$, we impose the $\theta$-geodesic gauge:

\[\Delta_{sc} (\theta(\dd{r})) = 0\]
where $\Delta_{sc}$ is the Laplacian with respect to the conformal Schwarzschild metric. This gauge condition is motivated by the structure of the linearized problem; it ensures that the linearized equations give a well-posed elliptic boundary value problem for $\theta$ in the function spaces used. Geometrically, this condition selects a canonical time slice within the family of initial data sets having the same Bartnik data.

\subsubsection*{Main Result and Key Insights}

\vv

The main theorem of this paper is as follows (refer to section \ref{functionspaces-sec} for the definition of the spaces and section \ref{mainthm-sec} for the precise formulation of the main theorem):

\begin{mainthm}
Let $M:= \R^3\setminus B_{r_0}$ where $r_0>2m_0$ and $m_0>0$. Let $\delta \in (-1,-\frac{1}{2}]$ and $k\geq 5$. Let $(M, \gsc, \Pi_{\scc})$ be any $H^k$ initial data set (possibly non time-symmetric) in the Schwarzschild spacetime $(\R \times M, \gsptime_{sc})$ of mass $m_0$. Denote by $( \gammasc, \frac{1}{2} \trksc, tr_{\partial M}\Pi_{\gsc}, {\omega_{\partial M}}_{\gsc})$ the Bartnik data of $\partial M$ in $(M, \gsc, \Pi_{\scc})$. There exists a neighbourhood $\mathcal{U}$ of $( \gammasc, \frac{1}{2} \trksc, tr_{\partial M}\Pi_{\gsc}, {\omega_{\partial M}}_{\gsc})$ in $$ \mathcal{M}^{k+1}(S^2) \times H^k(S^2) \times H^k(S^2) \times \Omega^k(S^2) $$ such that for every $(\gamma, H, \mathrm{tr}_{\partial M} \Pi, \omega_{\partial M}) \in \mathcal{U}$, there exists a vacuum stationary spacetime $(\R\times M, \gsptime)$, unique up to isometry, close to the Schwarzschild spacetime metric on $\R \times M$ in a certain Banach space that contains an initial data set $(M, \gsp, \Pi)$ with Bartnik data $(\gamma, H, \mathrm{tr}_{\partial M} \Pi, \omega_{\partial M})$ on $\partial M$. More precisely, there exists a unique metric $g$, 1-form $\theta$, and a function $u$ on $M$ such that 
\begin{itemize}
\item $g$ can be written globally in the form $g = dr^2+g(r)$, where $r = \mathrm{dist}(\cdot, \partial M)+r_0$ and $g(r)$ is the induced metric on the level sets of $r$. 
\item $\theta$ satisfies $\Delta_{sc} (\theta(\dd{r})) = 0$ on $M$. 
\item The spacetime metric $\gsptime := -e^{2u} (dt+\theta)^2 + e^{-2u} g$ on $\mathcal{M}:= \R \times M$ satisfies Einstein's vacuum equations, i.e. $\mathrm{Ric}_{\gsptime} = 0$. 
\item The initial data set $(M, \gsp, \Pi)$ formed by the $\{t=0\}$ hypersurface in $(\mathcal{M}, \gsptime)$ satisfies the desired Bartnik boundary conditions, i.e. the Bartnik data of $\partial M$ in $(M, \gsp, \Pi)$ is $(\gamma, H, \mathrm{tr}_{\partial M} \Pi, \omega_{\partial M})$. 
\end{itemize}

\end{mainthm}

A key observation is that the linearized equations decouple; letting $(\tilde g, \tilde u, \tilde X, \tilde \theta)$ be the linearized quantities of $(g,u,X, \theta)$ (where $X$ is an artificial vector field introduced to handle apparent obstructions as in \cite{ahmed}), the equations for $(\tilde g, \tilde u, \tilde X)$ are identical to those in the static case and decouple from the equations for $\tilde \theta$, which satisfies an elliptic boundary value problem. The $\tilde \theta$ problem takes the form

\begin{equation} 
\begin{cases}
4 \partial_r u_{sc}  \left(  dr \wedge \sgsc d\tilde \theta \right) + d\sgsc d\tilde \theta = d\sigma, & \text{in $M$}\\
 \cancel{div} (\tilde \theta^T) -2(r_0-2m_0) \tilde \theta_r = h, & \text{on $\partial M$}\\
\left( \mathcal{L}_{\dd{r}} \tilde \theta \right)^T +\slashed d \tilde \theta_r - \frac{2(r_0-3m_0)}{r_0(r_0-2m_0)} \tilde \theta^T = \Lambda,& \text{on $\partial M$}
\end{cases}
\end{equation}

together with the gauge condition 

\begin{equation} 
\begin{cases}
\Delta_{sc}( \tilde \theta(\dd{r})) = 0, & \text{in $M$}
\end{cases}
\end{equation}

for a 1-form $\sigma$ on $M$, a function $h$ and a 1-form $\Lambda$  on $\partial M$.

 We prove that this boundary value problem for $\tilde \theta$ is well-posed in the spaces $\mathcal{A}^{(t,k)}_\delta(M)$ introduced in \cite{ahmed}. The analysis requires careful treatment of the spherical harmonic decomposition and estimates uniform in the angular momentum $\ell$.

\vv

This work represents a step toward understanding the Bartnik conjecture outside time symmetry and opens the door to further investigations of the stationary extension problem and its implications for quasi-local mass in general relativity.

\subsection*{Organization of The Paper}

In section \ref{stationary-sec}, we provide the necessary background on stationary vacuum spacetimes, the quotient formalism, and the precise formulation of Bartnik boundary data in the stationary setting. We also establish notation and review the relevant analytic framework. In section \ref{functionspaces-sec}, we introduce the weighted Sobolev and Bochner-measurable function spaces that will be used throughout the analysis, and recall some key properties and embedding results. In section \ref{mainthm-sec}, we state the main local well-posedness theorem near coordinate Schwarzschild spheres in $\{ t=0\}$ Schwarzschild slice, together with the precise function space setting. We also prove a universal Lorentz-boost transformation law for the Bartnik data under a change of hypersurface through $\partial M$, and deduce as a corollary the corresponding local well-poseness statement near Schwarzschild spheres on any spacelike graph hypersurface in Schwarzschild spacetime. In section \ref{reduction-sec}, we reformulate the stationary vacuum Einstein equations in a double geodesic gauge, reducing the system to a coupled elliptic–transport–boundary value problem suitable for analysis in our chosen function spaces. In sections \ref{proof-sec}, we present the proof of the main theorem, which proceeds via the implicit function theorem on Banach manifolds. This section includes the analysis of the linearized system, the decoupling of the stationary contributions, and uniform estimates for the associated boundary value problems. In particular, we analyze in \ref{bvp-sec} the new boundary value problem for the 1-form $\tilde \theta$, establish solvability in the adapted function spaces, and derive the necessary $\ell$-uniform estimates for the spherical harmonic components.

\subsubsection*{Acknowledgements} The author is grateful to Spyros Alexakis for his invaluable insights during many discussions, assistance in verifying the mathematics and thoughtful advice throughout the project. The author is also grateful to Stephen McCormick for his helpful discussions and recommendations. Last but not least, the author acknowledges support from Knut and Alice Wallenberg Foundation under grant KAW 2022.0285.

\section{Preliminaries} \label{pre-sec}

\subsection{Stationary Vacuum Solutions} \label{stationary-sec}

\vv

Let $(\mathcal{M}, \gsptime)$ be a globally hyperbolic spacetime with boundary, with Cauchy hypersurface $M=\R^3\setminus B_{r_0}$, where $r_0 >2m_0$, $m_0>0$. The spacetime $\mathcal{M}$ is then diffeomorphic to $\R \times M$ and the spacetime metric can be written in the coordinates $(t,x^1, x^2, x^3)$ as follows: 

\begin{equation}
\gsptime = -e^{2u} dt^2 + Y^idx_i dt + \gsp_{ij} dx^i dx^j 
\end{equation}
where $\gsp$ is the induced Riemannian metric on the $\{t=0\}$ Cauchy hypersurface, $u$ is a function on $\mathcal{M}$ and $Y$ is a vector field on $\mathcal{M}$ tangent to the level sets of $t$. The function $e^{2u}$ and the vector field $Y$ are called the lapse function and shift vector, respectively, in the $3+1$ formalism of general relativity. 

\vv

Suppose in addition that the spacetime $\mathcal{M}$ is stationary, i.e. it admits a timelike Killing vector field. Assuming without loss of generality that $\dd{t}$ is the timelike Killing vector field, we have that 

\begin{equation}
\mathcal{L}_{\dd{t}} u = 0, \quad \mathcal{L}_{\dd{t}} Y = 0, \quad \mathcal{L}_{\dd{t}} \gsp = 0
\end{equation} 

That is, $u(t,x) = u(x)$, $Y(t,x) = Y(x)$ and $\gsp_{ij}(t,x) = \gsp_{ij}(x)$, for $t\in \R$ and $x\in M$. Note that $\gsptime$ is static if and only if $Y = 0$ on $M$. 

\vv
\begin{defn}
We say that a stationary spacetime $(\mathcal{M}, \gsptime)$ is asymptotically flat of order $\eta>0$ if there exist coordinates $(x^1,x^2,x^3)$ for $M$ near infinity such that
  \begin{itemize}
 \item $\mathfrak{g}_{ij} -\delta_{ij} = \mathcal{O}_2(|x|^{-\eta})$
\item $Y^i = \mathcal{O}_2(|x|^{-\eta})$
\item $u = \mathcal{O}_2(|x|^{-\eta})$
 \end{itemize}
 where $|x|:=\sqrt{|x^1|^2+|x^2|^2+|x^3|^2}$
\end{defn}
Here we use the notation $f = \mathcal{O}_k(|x|^{-\eta})$ to mean that $|D^j f| \leq C_j |x|^{-\eta-j}$ for all $0 \leq j \leq k$ and some constants $C_j > 0$, where $D$ is the covariant derivative with respect to the Euclidean metric on $M$.
\vv
\begin{defn}
An initial data set on $M$ is a triple $(M, \gsp, \Pi)$, where $\gsp$ is a Riemannian metric on $M$ and $\Pi$ is a $(0,2)$ symmetric tensor on $M$. \\
We say that an initial data set lives in a spacetime $(\mathcal{M}, \gsptime)$ if there exists an isometric embedding of $(M,\gsp)$ in $(\mathcal{M}, \gsptime)$ such that $\Pi$ coincides with the second fundamental form of $M$ in $\mathcal{M}$. \\
We say that an initial data set $(M,\gsp,\Pi)$ is asymptotically flat of order $\eta>0$ if there exists coordinates $(x^1,x^2,x^3)$ for $M$ near infinity such that
  \begin{itemize}
 \item $\mathfrak{g}_{ij} -\delta_{ij} = \mathcal{O}_2(|x|^{-\eta})$
\item $\Pi_{ij} = \mathcal{O}_1(|x|^{-\eta-1})$
\end{itemize}
 where $|x|:=\sqrt{|x^1|^2+|x^2|^2+|x^3|^2}$
\end{defn}

It is straightforward to check that if a stationary spacetime $\mathcal{M}$ is asymptotically flat, then $(M, \gsp, \Pi)$, where $\gsp$ and $\Pi$ are the induced metric and second fundamental form of the $\{t = 0\}$ hypersurface, is an asymptotically flat initial data set of the same order.

\vv

We now introduce the quotient formalism for stationary spacetimes, which provides a more canonical description by factoring out the time direction. This perspective is discussed in \cite{beig-simon, anderson-stationary, kramer}.

Let $(\mathcal{M}, \gsptime)$ be a stationary spacetime with timelike Killing vector field $\dd{t}$. Define the following equivalence relation on $\mathcal{M}$: 

\[ \text{For $p,q\in \mathcal{M}$,  $p\sim q$ if there exists an integral curve of $\dd{t}$ connecting $p$ and $q$ } \]

The quotient manifold $\mathcal{M} / \sim$, also called the orbit space, can be identified with $M$ via the natural diffeomorphism from $M$ to $\mathcal{M} / \sim$ defined by
\[ (x_1,x_2,x_3) \mapsto [(0,x_1,x_2,x_3)]\]

For any vector field $X$ on $M$, we can define the unique vector field $\bar X$ on $\mathcal{M}$ orthogonal to $\dd{t}$ and satisfying $\mathcal{L}_{\dd{t}}\bar X = 0$. In other words, $\bar X$ is the unique time-independent lift of $X$ to $\mathcal{M}$ that is orthogonal to the Killing direction $\dd{t}$; in particular, it is given by
\[\bar X = X - \theta(X)\dd{t} \]

 We will call $\bar X$ the horizontal lift of $X$ to $\mathcal{M}$. Denote by $\mathcal{X}_{\dd{t}}(\mathcal{M})$ the space of all vector fields $\bar X$ on $\mathcal{M}$ orthogonal to $\dd{t}$ and satisfying $\mathcal{L}_{\dd{t}}(\bar X) = 0$. Denoting the quotient map by $\pi : \mathcal{M} \to M$, we have 

\begin{equation}
\pi_{*} (\bar X) = X
\end{equation}

In fact, $\pi_{*}$ induces a module-isomorphism from $\mathcal{X}_{\dd{t}}(\mathcal{M})$ to the space $\mathcal{X}(M)$ of vector fields on $M$.  

 \begin{defn} The quotient metric $g^{\mathcal{Q}}$ on $M$ is defined by 

\[\text{For vector fields $X,Y$ on $M$,  } \qquad  g^{\mathcal{Q}}(X,Y) := \gsptime(\bar X, \bar Y) \] 
where $\bar X, \bar Y$ are the horizontal lifts of $X,Y$ to $\mathcal{M}$. In other words, $g^{\mathcal{Q}}$ is the induced Riemannian metric on the horizontal distribution orthogonal to the orbits of the Killing field $\dd{t}$.
\end{defn}

 For the simplicity of future computations, we will write the quotient metric in the following way 
\[ g^{\mathcal{Q}} = e^{-2u} g\]

for some Riemannian metric $g$.

\vv

The spacetime metric $\gsptime$ in the quotient formalism can then be expressed in terms of the quotient metric in the following way (see \cite{anderson-stationary, beig-simon, kramer}): 

\begin{equation} \label{quotient metric}
\gsptime = -e^{2u} (dt+\theta)^2 + e^{-2u} g
\end{equation}

where $\theta$ is a 1-form on $M$ that is related to $u$ and $Y$ by 

\begin{equation}
\theta = -e^{-2u} Y^{\flat_{\gsp}}
\end{equation}
Note that $\theta$ and $g$ in equation \eqref{quotient metric} are extended to $\mathcal{M}$ to be independent of $t$ and trivial on the normal bundle of $\{0\} \times M$. 

\vv

The induced metric $\gsp$ on the $\{t=0\}$ hypersurface and the metric $g$ are related by 

\begin{equation} \label{gsp-g}
\gsp = e^{-2u}g - e^{2u} \theta \otimes \theta
\end{equation}
Note that for the hypersurface to be spacelike, one requires that 
\[1-e^{4u} |\theta|_g^2 >0\]
It follows that a spacetime $\mathcal{M}$ is asymptotically flat of order $\eta>0$ if and only if there exists a coordinate system $(x_1,x_2,x_3)$ of $M$ near infinity such that 
\begin{itemize}
 \item $g_{ij} -\delta_{ij} = \mathcal{O}_2(|x|^{-\eta})$
\item $\theta_i = \mathcal{O}_2(|x|^{-\eta})$
\item $u = \mathcal{O}_2(|x|^{-\eta})$
\end{itemize}

This formalism has two advantages: 
\begin{itemize}
\item The stationary vacuum equations in terms of $u$, $g$ and $\theta$ take a simpler form compared to the formulation using $u$, $\gsp$ and $Y$. (compare \eqref{ric=0} with \eqref{ric=0-1})
\item While a given stationary spacetime contains many distinct initial data sets corresponding to different time slices, they all share the same quotient metric $g$, up to isometry. Hence, the quotient formalism is more canonical. 
\end{itemize}

\begin{defn} \label{stationary-def}
Let $g$, $u$, $\theta$ be a Riemannian metric, a function, and a 1-form on $M$, respectively. We will refer to the quadruple $(M,g,u,\theta)^{(4)}$ as the stationary spacetime $\mathcal{M} := \R \times M$ with the spacetime metric $\gsptime$ given by 
\begin{equation}
\gsptime = -e^{2u} (dt+\theta)^2 + e^{-2u} g
\end{equation} 
where $u$, $\theta$ and $g$ are extended to $\mathcal{M}$ to be independent of $t$ and trivial on the normal bundle of $\{0\} \times M$. 
\end{defn}

Given a stationary spacetime $\mathcal{M}$ defined by $(M,g,u,\theta)^{(4)}$, the induced metric $\gsp$ and the second fundamental form $\Pi$ of the $\{t=0\}$ hypersurface in $\mathcal{M}$ define an initial data set $(M,\gsp, \Pi)$ living in $\mathcal{M}$. The second fundamental form $\Pi$ can be computed to be 

\begin{equation}
\Pi = \frac{N}{2}  \mathcal{L}_{\bar \theta^{\sharp_{\mathfrak{g}}}} \mathfrak{g}
\end{equation}

where 
\begin{equation}
N = \frac{\sqrt{1-e^{4u}|\theta|^2_g}}{e^u}, \quad \bar \theta = e^{2u} \theta
\end{equation}

Different initial data sets embedded in the same stationary spacetime $\mathcal{M}$ correspond to, up to isometry, different choices of spacelike hypersurfaces in the given spacetime, which can be described by graphs of functions $f : M \to \mathbb{R}$. Given a function $f:M \to \R$, we will denote the graph of $f$ in $\mathcal{M}$ by $M_f := \{ (f(x), x) \in \mathcal{M} : x\in M \}$. Similarly, denote by $\gsp_f$ and $\Pi_f$ the induced metric and second fundamental form of $M_f$ in $\mathcal{M}$. We identify $M_f$ with $M$ via the map $(f(x), x_1,x_2,x_3) \mapsto (x_1, x_2,x_3)$; in particular, after pushing forward $\gsp_f$ and $\Pi_f$ along that map, we have that $(M, \gsp_f, \Pi_f)$ is an initial data set on $M$. 

\vv

 The following lemma characterizes all initial date sets living in a given stationary spacetime $\mathcal{M}$.

\begin{lem}\label{f-prop}
Let $\mathcal{M}$ be a stationary spacetime defined by $(M,g,u,\theta)^{(4)}$ that is asymptotically flat of order $\eta>0$. Then:

\begin{enumerate}[label=(\alph*)]
\item Every asymptotically flat initial data set $(M, \gsp, \Pi)$ of order $\eta>0$ living in $\mathcal{M}$ can be realized as $(M_f, \gsp_f, \Pi_f)$  in $\mathcal{M}$ for some function $f \in \mathcal{C}^2(M)$ satisfying $df = \mathcal{O}_1(|x|^{-\eta})$. 
\item Given $f \in \mathcal{C}^2(M)$ satisfying $df = \mathcal{O}_1(|x|^{-\eta})$ and the initial data set $(M, \gsp_f, \Pi_f)$ in $\mathcal{M}$, define the function $t':= t - f(x)$ on $\mathcal{M}$. Then the spacetime metric $\gsptime$ can be written in the coordinates $(t',x_1,x_2,x_3)$ in the following form: 
\[ \gsptime = -e^{2u} (dt'+\theta +df)^2 + e^{-2u} g \]
In particular, the initial data set $(M, \gsp_f, \Pi_f)$ can be realized as the $\{t=0\}$ initial data set in a stationary spacetime isometric to $(\mathcal{M}, \gsptime)$ --- namely the one  defined by $(M,g,u,\theta+df)^{(4)}$.

\end{enumerate}
Furthermore, the above defines the following one-to-one correspondences: 

\[ \left\{ \begin{aligned} & \text{AF Initial data sets $(M, \gsp, \Pi)$}\\ & \text{of order $\eta>0$}  \\ &\text{living in $(\mathcal{M}, \gsptime)$}\\ &\text{and preserving $\partial M$} \end{aligned}  \right\}    \longleftrightarrow  \left\{ f\in \mathcal{C}^2(M) \,\middle\vert \,\,\, \begin{aligned} & df=\mathcal{O}_1(|x|^{-\eta}),\\ & f|_{\partial M} = 0 \end{aligned} \right\}  \longleftrightarrow \left\{ \theta+ df \mid f \in \mathcal{C}^2(M), df=\mathcal{O}_1(|x|^{-\eta})\}  \right\}  \]

\end{lem}
\begin{proof}
The proof follows from direct computation using the transformation properties of the metric under the diffeomorphism $\Phi_f: \mathcal{M} \to \mathcal{M}$ defined by 
\[\Phi_f(t,x_1,x_2,x_3) = (t-f(x), x_1, x_2, x_3)\]. 
\end{proof}

\vv

\begin{figure}[H]
    \centering
    \includegraphics[scale=0.13]{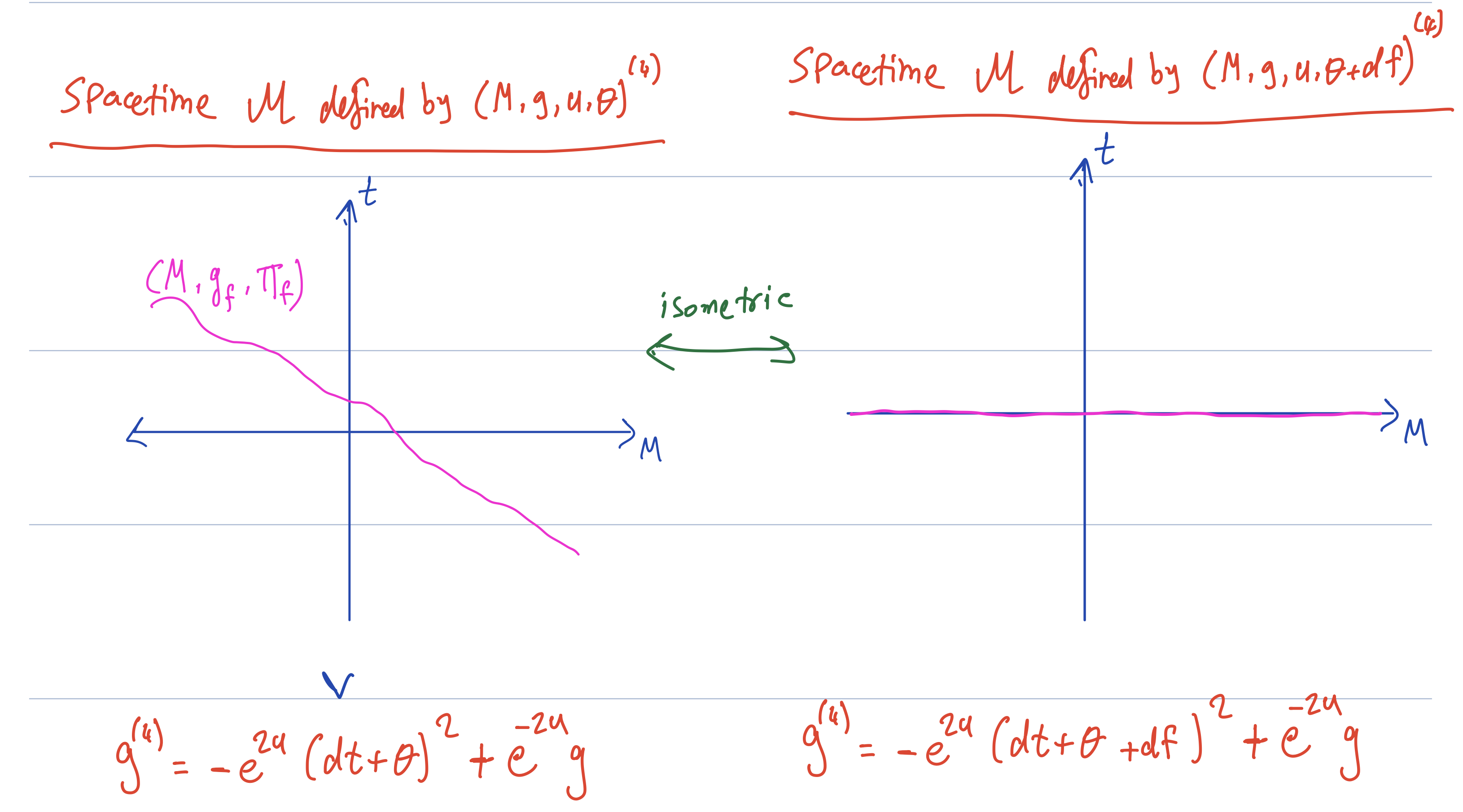}
    \caption{An illustration of an initial data set $(M,\gsp_f, \Pi_f)$ in the spacetime defined by $(M,g,u,\theta)^{(4)}$ is equivalent to the $\{t=0\}$ initial data set in the spacetime defined by $(M,g,u,\theta+df)^{(4)}$.   }
    \label{fig:domain}
\end{figure}

\vv

We will now define the Bartnik data of a given initial data set. 

\begin{defn} \label{bartnik-data-def}
Let $(M,\gsp, \Pi)$ be an initial data set on $M$. The Bartnik boundary data on $\partial M$ is $(\bartnik)$ where 
\begin{itemize}
\item $\gammao$ is the induced metric on $\partial M$ in $(M, \gsp)$. 
\item $\trko$ is the trace of the second fundamental form $K_{\mathfrak{g}}$ of $\partial M$ in $(M, \gsp)$. In particular, $\frac{1}{2} \trko$ is the mean curvature of $\partial M$ in $(M, \gsp)$. 
\item $\trkpio$ is the trace of the restriction $\Pi|_{\partial M}$ of the spacetime second fundamental form $\Pi$. 
\item $\omegao$ is the connection 1-form in the spacetime normal bundle of $\partial M$, defined by $\omegao(\cdot) = \Pi(\ngf, \cdot)$, where $\ngf$ is the inward unit normal on $\partial M$ in $(M,\gsp)$. 
\end{itemize}
\end{defn}

We can compute $\trkpio$ and $\omegao$ to be

\begin{equation}
\mathrm{tr}_{\partial M} \Pi = N \mathrm{div}_{\mathfrak{\gamma}} (\bar {\theta^{\sharp}}^T) - N\Big(\mathrm{tr}_{\partial M} K_{\mathfrak{g}}\Big) \bar \theta(\ngf) 
\end{equation}
\begin{equation}
\omegao(\cdot) =  \frac{N}{2}  \mathcal{L}_{\bar \theta^{\sharp_{\mathfrak{g}}}} \mathfrak{g} (\ngf, \cdot)
\end{equation}

\vv

Note that given a function $f$ on $M$, the initial data sets $(M, \gsp_f, \Pi_f)$ and $(M,\gsp, \Pi)$ living in the same stationary spacetime $\mathcal{M}$, might not necessarily have the same Bartnik boundary data on $\partial M$. The next lemma gives a necessary and sufficient condition on $f$ for the Bartnik boundary data to be preserved. 

\begin{lem} \label{f=0-lem}
Let $\mathcal{M}$ be a stationary spacetime. Let $f$ be any function on $M$. Then the Bartnik boundary data of the initial data sets $(M,\gsp, \Pi)$ and $(M, \gsp_f, \Pi_f)$ is preserved if and only if 
\begin{equation}
f_{\partial M} = 0, \quad df|_{\partial M} = 0
\end{equation}
\end{lem}
\begin{proof}
This follows from the transformation formulas for the induced metric and second fundamental form under the isometry $\Phi_f: \mathcal{M} \to \mathcal{M}$ defined by 
\[\Phi_f(t,x_1,x_2,x_3) = (t-f(x), x_1, x_2, x_3)\] See \cite{an-elliptic2} for the detailed calculation.
\end{proof}

\begin{defn}
Let $(\bartnik)$ be Bartnik boundary data on $\partial M$. Suppose that $\mathcal{M}$ be an asymptotically flat stationary spacetime of order $\eta>0$, defined by $(M, g, u ,\theta)^{(4)}$ satisfying Einstein's vacuum equations 
\begin{equation}
\mathrm{Ric}_{\gsptime} = 0
\end{equation}
such that $(\bartnik)$ coincides with the Bartnik data of the initial data set $(M, \gsp,\Pi)$ of the $\{t=0\}$ hypersurface in $\mathcal{M}$. If so, we say that the spacetime $(M, g, u ,\theta)^{(4)}$ is a stationary vacuum extension with Bartnik data $(\bartnik)$.  
\end{defn}

The Bartnik stationary metric extension conjecture is as follows. 

\begin{conj}
Let $(\bartnik)$ be Bartnik boundary data on $\partial M$. There exists a stationary vacuum extension with Bartnik data $(\bartnik)$ that is unique up to isometry.
\end{conj}

We will now rewrite Einstein's vacuum equations in equation \eqref{ric=0} in terms of the parameters $u,\theta$ and $g$. Suppose $\mathcal{M}$ is a stationary spacetime defined by $(M, g, u ,\theta)^{(4)}$. Let $\xi = \left( \dd{t} \right)^{\sharp} = -e^{2u} (dt+\theta)$ be the Killing 1-form. Define the twist form $\eta$ as the 1-form on $M$ defined by 

\begin{equation}
\eta = -\frac{1}{2} \star_{\gsptime} \xi \wedge d\xi = -\frac{1}{2} e^{4u} \star_{{g}} d\theta
\end{equation}

The form $\eta$ represents the obstruction to the integrability of the horizontal distribution in $T\mathcal{M}$, and so it measures the extent to which $\dd{t}$ fails to be orthogonal to the $\{t=0\}$ hypersurface. In particular, $\eta$ vanishes if and only if the metric $\gsptime$ is static. 

\vv

The stationary vacuum equations in the quotient formalism becomes as follows (see \cite{beig-simon}):
\begin{prop}
The stationary spacetime $(M, g, u ,\theta)^{(4)}$ satisfies Einstein's vacuum equations if and only if the following is satisfied on $M$. 

\begin{equation} \label{ric=0-1}
\mathrm{Ric}_g = 2 du \otimes du + 2e^{-4u} \eta \otimes \eta
\end{equation}
\begin{equation}
\Delta_g u = -2e^{-4u} |\eta|^2_{g}
\end{equation}
\begin{equation}
d\eta = 0
\end{equation}
\end{prop}

\vv

A result by Murchadha and Beig-Simon implies the following regarding the regularity and decay of stationary vacuum extensions (see \cite{murchadha, beig-simon}). 

\begin{prop} \label{decay prop}
Let $(M, g, u ,\theta)^{(4)}$ be an asymptotically flat stationary vacuum spacetime. Then $\gsptime$ is smooth away from the boundary and is Schwarzschildean near infinity. In particular, there exists coordinates near infinity in which 

\begin{equation}
\gsp_{ij} = \left( 1+ \frac{2m_0}{|x|}\right) \delta_{ij} + \mathcal{O}_2(|x|^{-2}), \quad u = -\frac{m_0}{|x|} + \mathcal{O}_2(|x|^{-2}), \quad \theta_i = \mathcal{O}_2(|x|^{-2})
\end{equation}
where $m_0$ is the ADM mass of the $\{t=0\}$ hypersurface $(M, \gsp)$. 
\end{prop}

\subsection*{The Geodesic Gauge}

Lemma \ref{f-prop} and Lemma \ref{f=0-lem} show that there is some freedom in choosing the stationary vacuum extension given the Bartnik data. In particular, two spacetimes $(M,g,u,\theta)^{(4)}$ and $(M,g',u',\theta')^{(4)}$ are stationary vacuum extensions of the same Bartnik data $(\bartnik)$ if they are isometric with isometry $\Phi$ of the form:  
 $$\Phi (t,x) = (t+f(x), \Psi(x)), \quad \text{for $(t,x) \in \mathcal{M}$},$$ where

\begin{itemize}

\item $\Psi$ is an isometry from $(M,g)$ to $(M,g')$ that fixes the boundary and is asymptotic to the identity map on $M$ so that $|\Psi(x) - \mathrm{Id}(x)| = \mathcal{O}_2(|x|^{-1})$.  
\item $f$ is a function on $M$ satisfying 
\begin{equation} \label{conditions for f} df = \mathcal{O}_1(|x|^{-\eta}), \quad f_{\partial M} = 0, \quad df|_{\partial M} = 0 \end{equation}
where $\eta>0$ is the decay rate of both spacetimes. 
 \end{itemize}
 
 If so, then $g = \Psi^* g'$, $u = \Psi^{*} u'$, and $\theta = \Psi^{*}\theta'+df$.

\vv

To pick a unique stationary spacetime $(M,g,u,\theta)^{(4)}$, we will choose a gauge for $g$, i.e. choose a particular element in the isometry class of $g$ in which the boundary and the asymptotics at infinity are preserved. Then, we will choose a gauge for $\theta$ i.e. choose a particular 1-form $\theta_{gauge}$ on $M$ among all the 1-forms $\theta+df$ where $f$ satisfies the conditions in \eqref{conditions for f}.

\vv

We will begin with choosing the gauge for $g$. We recall the following proposition from \cite{ahmed}.

\begin{prop} \label{gauge}
There exists $\tau'=\tau'(n,m_0)>0$ small enough such that the following is true for any $0<\tau<\tau'$.\\
If an asymptotically flat metric $g$ on $M$ of order $\eta>0$ satisfies in Cartesian coordinates
\begin{equation} \label{smallness}
|x|^{\eta} |g - g_{sc}| + |x|^{\eta+1}|\partial g - \partial g_{sc}| + |x|^{\eta+2}| \partial\partial g - \partial\partial g_{sc}| < \tau
\end{equation}
where $|\cdot|$ is with respect to the Euclidean metric $\delta$,  then:   

\begin{enumerate}
\item The affine parameter $r(\cdot) = \text{dist}_g(\partial M,\cdot)+r_0$ is differentiable everywhere on $M\setminus \partial M$ and defines a global radial foliation with leaves $S_r$ diffeomorphic to $S^2$. Moreover, given a coordinate system $(x^1,x^2,x^3)$ near infinity in which the metric satisfied $g_{ij} -\delta_{ij} = \mathcal{O}_2(|x|^{-\eta})$, we have that $r$ and $|x|$ are comparable in the sense that  
\begin{equation}
    C^{-1} |x| \leq r \leq C |x|
\end{equation}
for some constant $C>0$. 

\item With respect to this foliation, we have 
\begin{equation}  \label{trK-hatK-inf}
trK = \frac{2}{r} + \mathcal{O}_1(r^{-1-\eta}), \qquad |\hat K| = \mathcal{O}_1(r^{-1-\eta})
\end{equation}
where $K = \text{Hess}(r)$ is the second fundamental form on the leaves $S_r$, $trK$ is the trace of $K$, and $\hat K$ is the traceless part of $K$. 
\item There exists a unique diffeomorphism $\Phi: M \to [r_0, \infty) \times S^2 $ such that $\Phi|_{\partial M} = Id_{S^2}$, $r (\cdot) = \pi_r \circ \Phi (\cdot)$ where $\pi_r$ is the projection onto the first coordinate,  and $\Phi_* g = dr^2 + {\gamma_g}_r$ where ${\gamma_g}_r$ is the push forward of the induced metric on $S_r$. 
\end{enumerate}
\end{prop}

This allows us to globally express $g$ in the stationary vacuum extension $(M, g, u ,\theta)^{(4)}$ near the Schwarzschild solution in geodesic coordinates as follows: 
\[g =  dr^2 + g(r), \quad \text{where $g(r)$ is the induced metric on $S_r$}\]
and, hence, determines the gauge that we will use for $g$, which we will call the $g$-geodesic gauge. 

\vv

Fix Bartnik boundary data $(\bartnik)$ and let $(M, g, u ,\theta)^{(4)}$ be a stationary vacuum extension realizing this Bartnik data. Proposition \ref{f-prop} shows that there is a one to one correspondence between stationary vacuum extensions $(M,g,u,\theta')$ realizing the Bartnik data $(\bartnik)$ and functions $f$ on $M$ satisfying 
\[ f|_{\partial M} = 0,\quad  df|_{\partial M} =0, \quad df = O(r^{-\eta}) \]

In light of proposition \ref{f-prop}, this one-to-one correspondence is given by 
\[f \mapsto (M, g, u ,\theta+df)^{(4)}\]

We now fix the remaining gauge freedom for $\theta$ by imposing a canonical condition. 

\begin{prop} \label{uniquef-prop}
Let $(M, g, u ,\theta)^{(4)}$ be a stationary vacuum extension. Suppose $g$ can be globally expressed in geodesic coordinates: 
\[g = dr^2 + g(r)\]
There exists a unique $f$ satisfying 
\[ f|_{\partial M} = 0,\quad  df|_{\partial M} = 0, \quad df = \mathcal{O}_1(r^{-\eta}) \]
such that the 1-form $\theta_f:= \theta+df$ satisfies 
\begin{equation}
\Delta_{g_{sc}} \left( \theta_f(\dd{r})\right)=0
\end{equation}
\end{prop}
\begin{proof}
It suffices to prove that there exists a unique $f$ on $M$ satisfying: 

\begin{equation}
\begin{cases}
&\Delta_{g_{sc}} \partial_r f = -\Delta_{g_{sc}}(\theta(\dd{r})) \\
&df = \mathcal{O}_1(r^{-\eta})\\
& f|_{\partial M} = 0\\
&\partial_r f|_{\partial M} = 0
\end{cases}
\end{equation}

It follows by solvability of the Laplacian on asymptotically flat manifolds (see \cite{Bartnik1}) that there exists a unique smooth function $h = \mathcal{O}_1(r^{-\eta})$ satisfying

\begin{equation}
\begin{cases}
&\Delta_{g_{sc}} h = -\Delta_{g_{sc}}(\theta(\dd{r}))\\
&h|_{\partial M} = 0
\end{cases}
\end{equation}

Note that we used the fact that $\theta(\dd{r}) = \mathcal{O}_2(r^{-2})$ according to proposition \ref{decay prop}. 

\vv

By integrating $h$ in $r$, it is clear that there exists a unique function $f$ satisfying $\partial_r f = h$ and $f|_{\partial M} = 0$.
\end{proof}

\begin{figure}[H]
    \centering
    \includegraphics[scale=0.12]{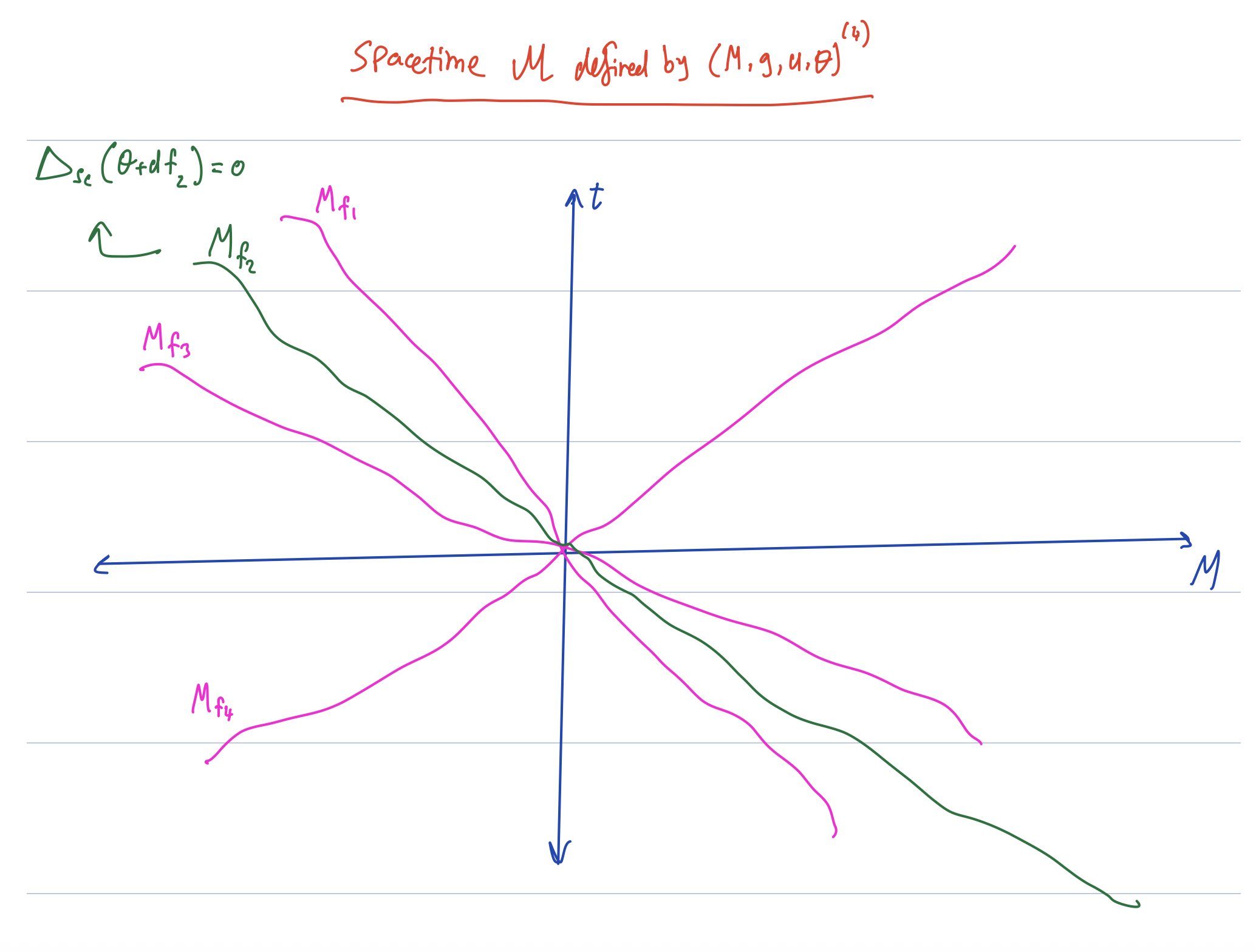}
    \caption{Illustration of the gauge freedom for initial data sets in a stationary spacetime defined by $(M, g, u, \theta)^{(4)}$. Each colored curve $M_{f_i}$ represents a different spacelike hypersurface given by the graph $t = f_i(x)$, where the functions $f_i$ satisfy the boundary conditions $f_i|{\partial M} = 0$, $df_i|{\partial M} = 0$ so that all initial data sets realize the same Bartnik boundary data on $\partial M$. Proposition \ref{uniquef-prop} shows that there exists a unique choice of hypersurface—depicted in green, $M_{f_2}$—for which the corresponding $1$-form $\theta + df_2$ satisfies the $\theta$-geodesic gauge condition $\Delta_{sc}\Big((\theta + df_2)(\dd{r})\Big) = 0$. This means that for any prescribed Bartnik data, there is a unique representative in the gauge where $\theta$ is determined by this elliptic condition and the corresponding initial data slice is then the ${t=0}$ hypersurface in the isometric stationary spacetime defined by $(M,g,u,\theta +df_2)^{(4)}$. Thus, fixing the $\theta$-gauge is equivalent to selecting a unique initial data slice among all those realizing the same boundary data in the stationary spacetime.}
    \label{fig:domain}
\end{figure}

\begin{defn} \label{theta-gauge-defn}
Let $(M,g,u,\theta)^{(4)}$ be a stationary spacetime in which $g$ is in the $g$-geodesic gauge. We say that $\theta$ is in the $\theta$-geodesic gauge if it satisfies

\begin{equation}
\Delta_{g_{sc}} \left( \theta(\dd{r}) \right) = 0
\end{equation} 
\end{defn}

With these gauge choices, we have completely fixed all diffeomorphism freedom in our problem. Specifically: 
\begin{itemize}
\item The $g$-geodesic gauge $g = dr^2 + g(r)$ fixes the spatial diffeomoprhisms that preserve the foliation structure. 
\item The $\theta$-geodesic gauge determines the additive gauge freedom in $\theta$. 
\end{itemize}

Since we have exhausted all available gauge freedom, we expect the stationary vacuum extension problem with prescribed Bartnik data to have at most one solution in these gauges, leading to our uniqueness result.

\subsection{Function Spaces} \label{functionspaces-sec}

In this section, we define the function spaces that we will be using. Fix $k \in \Z_{\geq0}$ and $\delta \in \R$. From here onwards, we will identify $M$ with the space $[r_0,\infty) \times S^2$. 

\begin{defn}
We define the weighted Sobolev space $H^k_{\delta}([r_0,\infty))$ with weight $\delta$ to be the space of all functions $f \in H^k_{loc}([r_0,\infty))$ such that $\normmmH{f}{k}{\delta}<\infty$, where 

\begin{equation}
\normmmH{f}{t}{\delta}^2 = \sum_{k'=0}^k \int_{r_0}^{\infty} r^{-2\delta-1+2k'} \left(f^{(k')}(r)\right)^2 dr
\end{equation}
We will also denote the space $H^k_{\delta}([r_0,\infty))$ by $L^2_{\delta}([r_0,\infty))$ when $k=0$.
\end{defn}

\begin{defn}
We define the weighted space $C^k_{\delta}([r_0,\infty))$ with weight $\delta$ to be the space of all functions $f \in C^k([r_0,\infty))$ such that $\normmmC{f}{k}{\delta}<\infty$, where 

\begin{equation}
\normmmC{f}{k}{\delta}^2 = \sum_{k'=0}^k \sup\left( r^{-2\delta +2k'} ({f}^{(k')}(r))^2 \right)
\end{equation}
\end{defn}

\begin{defn} 
 We define the weighted Sobolev space $H^k_{\delta} (M)$ with weight $\delta$ to be the space of all functions $u$ in $H^k_{loc}(M)$ such that $\norm{u}_{k,\delta}< \infty$ respectively, where 
\begin{equation} \label{norm1}
 \norm{u}_{k, \delta} = \sum^k_{l = 0} \left\{ 
\int_M \left( |D^l u| \cdot r^{l -\delta} \right)^2 r^{-3} dV \right\}^
\frac{1}{2}
\end{equation}
where $r = |x|$, $D$ is the connection with respect to the Euclidean metric on $M$, and $dV$ is the Euclidean volume form on $M$. We will also denote the space $H^k_{\delta}(M)$ by $L^2_{\delta}(M)$ when $k=0$.\\
\end{defn}

\begin{defn} \label{def of Xkdelta}
\noindent We define the space $\mathcal{X}^k_{\delta}(M)$ to be the space of vector fields $X$ on $M$ with components $X^i := X(x^i)$ in $H^k_{\delta}(M)$, where $(x^1,x^2,x^3)$ is the standard cartesian coordinates. The norm we use is 
\begin{equation}
\norm{X}_{k,\delta}:= \sum_{l=0}^k \norm{|D^l X|}_{0,\delta-l}
\end{equation}
\end{defn}

\vv


\vv

\begin{defn} 
Let $H^k(S^2)$ be the usual $L^2$ space, when $k=0$, and Sobolov space, when $k\geq 1$, on $(S^2, \gamma_{\mathbb{S}^2})$. Let $\mathcal{M}^k(S^2)$ and $\mathcal{H}^k(S^2)$ be the space of metrics on $S^2$ and symmetric tensors on $S^2$, respectively, with components in $H^k(S^2)$. The norm we will use is as follows:
\begin{equation}
\norm{h}^2_{\mathcal{H}^k(S^2)} := \sum_{l=0}^k \norm{|\slashed D^l h|}^2_{L^2(S^2)}
\end{equation}
where $\slashed D$ is the covariant derivative on $S^2$ with respect $\gamma_{\mathbb{S}^2}$. \\
\end{defn}
\begin{defn}
\noindent Let $\Omega^k(S^2)$ be the space of 1-forms on $S^2$ with components in $H^k(S^2)$. The norm used on this space is as follows:
\begin{equation}
\norm{\omega}^2_{\Omega^k(S^2)} := \sum_{l=0}^k \norm{|\slashed D^l \omega|}^2_{L^2(S^2)} 
\end{equation}  
\end{defn}
\vv

\begin{defn} Let $t \in \Z_{\geq 0}$. We define the space $H^t_{\delta}\left([r_0,\infty); H^k(S^2)\right)$ to be the space of functions $u$ in $H^t_{loc}\left([r_0,\infty); H^k(S^2)\right)$ such that $\normmH{u}{t}{k}{\delta} <\infty$, where 

\begin{equation}
\normmH{u}{t}{k}{\delta}^2:= \sum_{t'=0}^t\int_{r_0}^{\infty} r^{-2\delta-1+2t'}   \norm{ \partial_r^{(t')} u(r)}^2_{H^k(S^2)} dr
\end{equation}

We also define the space $\CtoH{t}{\delta}{k}$ to be the space of continuous $H^k(S^2)$-valued functions $u$ on $[r_0,\infty)$ such that $\normmC{u}{t}{k}{\delta} <\infty$, where 

\begin{equation}
\normmC{u}{t}{k}{\delta}^2:= \sum_{t'=0}^t \sup_{r\geq r_0} \left( r^{-2\delta+2t'}   \norm{ \partial_r^{(t')} u(r)}^2_{H^k(S^2)}\right)
\end{equation}
\end{defn}
\vv

\vv

We then define the space $H^t_{\delta}\left([r_0,\infty); \mathcal{M}^k(S^2)\right)$ and $H^t_{\delta}\left([r_0,\infty); \mathcal{H}^k(S^2)\right)$ similarly to the above with norm
\begin{equation}
\normmH{h}{t}{k}{\delta}^2 := \sum_{t'=0}^t\int_{r_0}^{\infty} r^{-2\delta-1+2t'} \norm{\partial_r^{(t')} h(r)}^2_{\mathcal{H}^k(S^2)} dr
\end{equation}

\vv

\begin{defn}   \label{space-of-metrics}
Define $\mathcal{M}^k_{\delta}(M)$ to be the space of metrics on $[r_0,\infty) \times S^2$ of the form $dr^2+ g(r)$ where $g(r) = r^2(\gamma_{\infty} + h(r))$, $\gamma_{\infty} \in \mathcal{M}^k(S^2)$, and $h \in \HtoHH{2}{\delta}{k}$. The space $\mathcal{M}^k_{\delta}(M)$ can be naturally identified with an open subset of the Banach space $\mathcal{H}^k(S^2) \oplus \HtoHH{2}{\delta}{k}$. This makes $\mathcal{M}^k_{\delta}(M)$ an open Banach submanifold of $\mathcal{H}^k(S^2) \oplus \HtoHH{2}{\delta}{k}$ and, in particular, a Banach manifold. Given $g_0 \in \mathcal{M}^k_{\delta}(M)$, the tangent space $T_{g_0}M^k_{\delta}$ is isomorphic to the space of tensors $\tilde g$ of the form $\tilde g = r^2(\tilde \gamma_{\infty}+ \tilde h(r))$, where $\tilde \gamma_{\infty} \in \mathcal{H}^k(S^2)$ and $\tilde h \in  H^2_{\delta}\left([r_0,\infty); \mathcal{H}^k(S^2)\right)$, equipped with the norm 

\begin{equation}
\norm{\tilde g}_{\mathcal{M}^k_{\delta}} := \norm{\tilde \gamma_{\infty}}_{\mathcal{H}^k(S^2)} + \normmH{\tilde h}{2}{k}{\delta} 
\end{equation}

\end{defn}

\vv
\begin{defn} \label{spaces for u}
Let $t\geq 0$. Denote by $\AH{t}{k}{\delta}(M)$ and $\AC{t}{k}{\delta}(M)$ the spaces

\begin{equation}
\AH{t}{k}{\delta}(M) := \bigcap_{t'=0}^t \HtoH{t'}{\delta}{k-t'} , \qquad \AC{t}{k}{\delta}(M) :=  \bigcap_{t'=0}^t \CtoH{t'}{\delta}{k-t'}
\end{equation}
equipped with the norms

\begin{equation}
\norm{u}^2_{\AH{t}{k}{\delta}} :=\max_{0\leq t'\leq t}  \normmH{u}{t'}{k-t'}{\delta}^2 , \qquad \norm{u}^2_{\AC{t}{k}{\delta}} := \max_{0\leq t'\leq t}  \normmC{u}{t'}{k-t'}{\delta}^2 
\end{equation}

Note that 
\[ u \in \AH{t}{k}{\delta}(M) \quad \Longleftrightarrow \quad \text{for every $0\leq t'\leq t$,   } \,\,  \partial_r^{(t')} u \in \LtoH{\delta-t'}{k-t'} \]
\[ u \in \AC{t}{k}{\delta}(M) \quad \Longleftrightarrow \quad \text{for every $0\leq t'\leq t$, } \,\, \partial_r^{(t')} u \in \CtoH{0}{\delta-t'}{k-t'} \]
Denote the intersection of these spaces by $\AHC{t}{k}{\delta}(M)$ defined by
\begin{equation*}
\AHC{t}{k}{\delta}(M) := \AH{t}{k}{\delta}(M) \bigcap \AC{t}{k}{\delta}(M)
\end{equation*}
equipped with the norm 

\begin{equation}
\norm{u}_{\mathcal{A}^{(t,k)}_{\delta}}^2:= \max_{0\leq t'\leq t} ( \normmH{u}{t'}{k-t'}{\delta}^2 + \normmC{u}{t'}{k-t'}{\delta}^2 )
\end{equation}

\begin{defn}
We define ${\Omega}^{(t,k)}_{\delta}(M)$ to be the space of $1$-forms on $M$ with components in $\AHC{t}{k}{\delta}(M)$ with respect to the standard cartesian coordinates. 
\end{defn}

\begin{defn}
We define ${\Omega_{\mathcal{G}}}^{(t,k)}_{\delta}(M)$ to be the space of $1$-forms $\theta \in {\Omega}^{(t,k)}_{\delta}(M)$ satisfying the $\theta$-geodesic gauge condition introduced in proposition \ref{theta-gauge-defn}: 
\begin{equation}
\Delta_{g_{sc}} \left( \theta(\dd{r}) \right) = 0
\end{equation}
\end{defn}

\end{defn}

\vvv

Denote by $\slashed d$ the exterior derivative on $S^2$ and $\slashed \star$ the Hodge star operator on $(S^2, \gamma_{\mathbb{S}^2})$. In the next proposition, we list some important properties regarding the Hodge decomposition of 1-forms on $S^2$ and $M$ that will be repeatedly used in the rest of the paper.

\begin{prop}\label{hodge-prop} Let $k\in \Z_{\geq 0}$. \
\begin{enumerate}[label=\textbf{(\alph*)}]

\item A 1-form $\omega$ on $S^2$ lies in $\Omega^k(S^2)$ if and only if it can be written in the form

\begin{equation}
\omega = \slashed d a + \slashed \star \slashed d b
\end{equation}
where $a, b \in H^{k+1}(S^2)$. 

\item A 1-form $\theta$ on $M$ lies in ${\Omega}^{(t,k)}_{\delta}(M)$ if and only if it can be written in the form 

\begin{equation}
\theta = \slashed d a + \slashed \star \slashed d b + c dr
\end{equation}
where $ a,  b \in \AHC{t}{k+1}{\delta+1}(M)$ and $c \in \AHC{t}{k}{\delta}(M)$. 
\end{enumerate}

\end{prop}
\begin{proof}
(a) follows from standard results on the Hodge decomposition of 1-forms on closed manifolds (see \cite{hodge}). 

\vv

To prove (b), fix a 1-form $\theta$ on $M$. Then for each $r \in [r_0,\infty)$, the projection $\theta^{T_r}$ of $\theta$ on the $S_r :=\{r\} \times S^2$ admits the decomposition 
\begin{equation}
\theta^{T_r} = \slashed d a(r) + \slashed \star \slashed d b(r)
\end{equation}
for functions $a(r) , b(r)$ on $S^2$. Letting $c = \theta(\dd{r})$, we have that for each $r \in [r_0,\infty)$

\begin{equation}
\theta|_{S_r} = \theta^{T_r} + c(r) dr = \slashed d a(r) + \slashed \star \slashed d b(r) + c(r) dr
\end{equation}
where $\theta|_{S_r}$ is the restriction of $\theta$ on $S_r$. Equivalently, we have 
\begin{equation}
\theta = \slashed d a + \slashed \star \slashed d b + c dr
\end{equation}
for functions $a,b,c$ on $M$. 

\vv

 By (a), $\theta^{T_r} \in \Omega^k(S^2)$ if and only if $a(r), b(r) \in H^{k+1}(S^2)$. Furthermore, we compute 
 
 \begin{equation}
 \slashed d \slashed \star \theta^{T_r} = (-\slashed \Delta a ) d \sigma_{\gamma_{\mathbb{S}^2}}, \qquad  \slashed \star \slashed d \theta^{T_r} = -\slashed \Delta b 
 \end{equation}
 where $\slashed \Delta$ is the Laplacian and $d\sigma_{\mathbb{S}^2}$ is the volume form on $(S^2, \gamma_{\mathbb{S}^2})$. After integrating by parts on $(S_r, r^2 \gamma_{\mathbb{S}^2})$, we then get 
\begin{equation}
\int_{S_r} (|\slashed \nabla_{S_r} \theta^{T_r}|_{S_r}^2 + \frac{1}{r^2} |\theta^{T_r}|_{S_r}^2) d\sigma_{S_r} = \int_{S_r} (|\slashed \Delta_{S_r} a(r)|^2 + |\slashed \Delta_{S_r} b(r)|^2)d\sigma_{S_r}
\end{equation}

where $\slashed \nabla_{S_r}$, $\slashed \Delta_{S_r}$, $|\cdot|_{S_r}$ and $d\sigma_{S_r}$ are with respect to the round metric $r^2 \gamma_{\mathbb{S}^2}$ on $S_r$. In particular, we deduce that 

\begin{equation}
\begin{aligned}
\frac{1}{C} \left( \int_{r_0}^{\infty} r^{-2\delta-3} (\norm{a(r)}^2_{H^2(S^2)} + \norm{b(r)}^2_{H^2(S^2)}) dr\right)\\ \leq \int_{r_0}^{\infty} r^{-2\delta-1} \norm{\theta^{T_r}}^2_{\Omega^1(S^2)} dr \leq \qquad \qquad  \\ C \left( \int_{r_0}^{\infty} r^{-2\delta-3} (\norm{a(r)}^2_{H^2(S^2)} + \norm{b(r)}^2_{H^2(S^2)}) dr\right)
\end{aligned}
\end{equation}

for some constant $C>0$ depending only on $\delta$ and $r_0$. 
 It then follows directly that $\theta \in {\Omega}^{(0,1)}_{\delta}(M)$ if and only if $a,b \in \AHC{0}{2}{\delta+1}(M)$ and $c \in \AHC{0}{1}{\delta}(M)$. Using similar arguments, we conclude that $\theta \in {\Omega}^{(t,k)}_{\delta}(M)$ if and only if $a,b \in \AHC{t}{k+1}{\delta+1}(M)$ and $c \in \AHC{t}{k}{\delta}(M)$

\end{proof}

\subsection{The Main Theorem} \label{mainthm-sec}

This section records the main local well-posedness result for the stationary vacuum extension
problem in our fixed gauges, and explains how it immediately implies a corresponding local
well-posedness statement for Bartnik data prescribed on \emph{boosted} hypersurfaces.  The rest
of the paper is devoted to proving Theorem~\ref{main-thm} below.  The boosted statement is a
formal corollary once one observes that changing the spacelike hypersurface through $\partial M$
amounts to a Lorentz boost in the spacetime normal bundle of $\partial M$, and hence the Bartnik
data transform by an explicit, universal formula.

\begin{thm}\label{main-thm}
Let $M:= \R^3\setminus B_{r_0}$ where $r_0>2m_0$ and $m_0>0$. Let $\delta \in (-1,-\frac{1}{2}]$ and $k\geq 5$. Let $( \gammasc, \frac{1}{2} \trksc, 0, 0)$ be the Bartnik data on $\partial M$ in the $\{t=0\}$ hypersurface of Schwarzschild spacetime. There exists a neighbourhood $\mathcal{U}$ of $( \gammasc, \frac{1}{2} \trksc, 0, 0)$ in $$ \mathcal{M}^{k+1}(S^2) \times H^k(S^2) \times H^k(S^2) \times \Omega^k(S^2) $$ and a unique $\mathcal{C}^1$ map ${\bf H}^{\mathcal{G}}: (\bartnik) \mapsto (g, u,\theta)$ on $\mathcal{U}$ into $$ \mathcal{M}^k_{\delta}(M) \times \mathcal{A}^{(2,k+1)}_{\delta}(M) \times  {\Omega_{\mathcal{G}}}^{(2,k+1)}_{\delta}(M) $$
such that $(M,g,u,\theta)^{(4)}$ is a stationary vacuum extension with Bartnik data $(\bartnik)$. 
\end{thm}

\vv

In particular, given Bartnik data $(\bartnik)$ in $\mathcal{U}$, there exists a unique asymptotically flat stationary spacetime $(\mathcal{M}, \gsptime)$ defined by $(M,g,u,\theta)^{(4)}$ that is close to the Schwarzschild solution $(\mathcal{M}, \gsptime_{sc})$ in which 
\begin{itemize}
\item $(\mathcal{M}, \gsptime)$ satisfies Einstein's vacuum equations.  
\item The Bartnik data of the boundary $\partial M$ in the initial data set $(M, \gsp, \Pi)$ of the $\{t=0\}$ hypersurface in $(\mathcal{M}, \gsptime)$ coincides with $(\bartnik)$. 
\item The metric $g$ satisfies the $g$-geodesic gauge, i.e. $g$ can be globally written in the form $g = dr^2 + g(r)$. 
\item The 1-form $\theta$ satisfies the $\theta$-geodesic gauge, i.e. $\theta$ satisfies $\Delta_{sc} \left( \theta(\dd{r}) \right)=0$. 
\end{itemize}

\vv

We now deduce a local well-posedness statement near \emph{any} Schwarzschild sphere sitting inside a
spacelike graph hypersurface $\{t=f(x)\}$ in the Schwarzschild spacetime preserving the intrinsic geometry of $\partial M$ (i.e. $f|_{\partial M} = 0$).  The key point is that the
change of hypersurface induces a Lorentz boost in the spacetime normal bundle of $\partial M$; the
induced Bartnik data transform by a universal formula depending only on the corresponding rapidity.

\begin{lem}[Boost transformation of Bartnik data]\label{lem:boost-bartnik}
Let $(\mathcal{M}, \gsptime)$ be a stationary spacetime defined by $(M,g,u,\theta)^{(4)}$. Let $f$ be a $\mathcal{C}^2$ function on $M$ vanishing on the boundary $\partial M$ so that its graph is spacelike. Denote by $(M, \gsp, \Pi)$ and $(M_f, \gsp_f, \Pi_f)$ the initial data set defined by $\{t=0\}$ and $\{t=f(x)\}$ in $(\mathcal{M}, \gsptime)$. Denote by $(\bartnik)$ and $(\bartnikf)$ the Bartnik data of $\partial M$ in $(M,\gsp, \Pi)$ and $(M_f, \gsp_f, \Pi_f)$. 

There exists a unique $\mathcal{C}^1$ function $\varphi$ on $\partial M$, called the rapidity, such that

\begin{itemize}
\item ${\gammao}_f = \gammao$
\item ${\trko}_f = \cosh (\varphi) \trko +\sinh(\varphi) \trkpio$
\item ${\trkpio}_f = \sinh (\varphi) \trko + \cosh(\varphi) \trkpio$
\item ${\omegao}_f = \omegao + \slashed d \varphi$
\end{itemize}

The second and third Bartnik data can equivalently be written in matrix form as follows: 

\begin{equation} \label{bartnik change} \begin{pmatrix}
{\trko}_f\\[2pt]
{\trkpio}_f
\end{pmatrix}
=
\begin{pmatrix}
\cosh(\varphi) & \sinh(\varphi)\\
\sinh(\varphi) & \cosh(\varphi)
\end{pmatrix}
\begin{pmatrix}
\trko\\[2pt]
\trkpio
\end{pmatrix}\end{equation}

\vv

In particular, the map describing how the Bartnik data changes, defined by
\[(\bartnik) \mapsto (\bartnikf),\]
depends only on the rapidity $\varphi$, and hence we denote it by $\mathcal{B}_{\varphi}$. Furthermore, if $\varphi \in H^{k+1}(S^2)$, then $\mathcal{B}_{\varphi}$ is a $\mathcal{C}^1$ diffeomorphism of $$\mathcal{M}^{k+1}(S^2)\times H^k(S^2)\times H^k(S^2)\times \Omega^k(S^2)$$ with inverse $\mathcal{B}_{-\varphi}$. 
\end{lem}

\begin{proof}
Since $f$ vanishes on the boundary, the induced metric on $\partial M$ is preserved, i.e. ${\gammao}_f = \gammao$. 

\vv

Denote by $\bf{T}$ and $\bf{n}$ the future-oriented unit normal to $(M, \gsp)$ in $(\mathcal{M}, \gsptime)$ on $\partial M$ and the unit normal to $\partial M$ in $(M, \gsp)$, respectively. Define similarly the vector fields ${\bf{T}}_f$ and ${\bf{n}}_f$ on $\partial M$ for the $(M_f, \gsp_f)$ hypersurface. 
For each $p \in \partial M$, the normal space $N_p \partial M \subset T_p \mathcal{M}$ is a 2-dimensional Lorentzian plane, and the ordered pairs ${(\bf{T}, \bf{n})}_p$ and $({\bf{T}}_f, {\bf{n}}_f)_p$ are both positively oriented orthonormal bases of $N_p\partial M$. They differ by a unique element of $SO^{+}(1,1)$, implying that for a unique number $\varphi(p)$, we have

\begin{equation} \label{Tf nf}
{\bf{T}}_f = \cosh(\varphi) \bf{T} + \sinh(\varphi) \bf{n}, \quad {\bf{n}}_f = \sinh(\varphi) \bf{T} + \cosh (\varphi) \bf{n}
\end{equation}

This defines the rapidity $\varphi$ as a function on $\partial M$. Since $f$ is $\mathcal{C}^2$, the orthonormal frame $({\bf{T}}_f, {\bf{n}}_f)$ is $\mathcal{C}^1$, implying that $\varphi \in \mathcal{C}^1(\partial M)$. In fact, the rapidity is determined from the normal derivative of $f$ according to the following formula: 

\begin{equation}\label{varphi f}
 N^{-1} \frac{{\bf n} (f)}{1-e^{2u} \theta ({\bf n}) {\bf n}(f)} = \tanh \varphi,
 \end{equation}

\vv

Denote by $\bf{H}$ the mean curvature vector field of $\partial M$ in $(\mathcal{M}, \gsptime)$, which depends only on the embedding of $\partial M$ in $\mathcal{M}$. Since $\bf{H}$ is a section of the normal bundle $N \partial M$, we can write it as follows: 

\[{\bf{H}} = -\trkpio {\bf T} + \trko {\bf n} = -{\trkpio}_f {\bf{T}}_f + {\trko}_f {\bf{n}}_f \]

The above directly implies equation \eqref{bartnik change}.  

\vv

Finally, for $\omegao$ and ${\omegao}_f$, we compute using equation \eqref{Tf nf} that for a vector field $X$ on $\partial M$,

\begin{align*}
{\omegao}_f  (X)&= \Pi_f ( X, {\bf{n}}_f) \\
&= \gsptime(\nabla_X {\bf{T}}_f, {\bf{n}}_f)\\ 
&= \gsptime \bigg(\nabla_X \Big( \cosh(\varphi) \bf{T} + \sinh(\varphi) \bf{n} \Big), \sinh(\varphi) \bf{T} + \cosh (\varphi) \bf{n} \bigg)\\
&= \bigg( \sinh^2(\varphi) \gsptime({\bf T}, {\bf T}) + \cosh^2(\varphi) \gsptime({\bf n}, {\bf n}) \bigg) X(\varphi) + (\cosh^2(\varphi) - \sinh^2(\varphi)) \gsptime(\nabla_X {\bf T}, {\bf n})\\
&= X(\varphi) + \Pi(X, {\bf n})\\
&= (\omegao + \slashed d \varphi)(X)
\end{align*}

implying that 
\[{\omegao}_{f} = \omegao + \slashed d \varphi\]

as needed

\end{proof}

\vv

\begin{cor}[Local well-posedness near a boosted Schwarzschild slice]\label{cor:boosted-main}

Consider an arbitrary $\varphi$-boosted coordinate sphere in the Schwarzschild spacetime $(\mathcal{M}, {\gsptime}_{sc})$, where the rapidity $\varphi$ is in $H^{k+1}(S^2)$. There exists a neighbourhood $\mathcal{U}_{\varphi}$ of $\mathcal{B}_{\varphi}( \gammasc, \frac{1}{2} \trksc, 0, 0)$ in $$ \mathcal{M}^{k+1}(S^2) \times H^k(S^2) \times H^k(S^2) \times \Omega^k(S^2) $$ and a unique $\mathcal{C}^1$ map ${\bf H}^{\mathcal{G}}_{\varphi}: (\bartnik) \mapsto (g, u,\theta)$ on $\mathcal{U}_{\varphi}$ into $$ \mathcal{M}^k_{\delta}(M) \times \mathcal{A}^{(2,k+1)}_{\delta}(M) \times  {\Omega_{\mathcal{G}}}^{(2,k+1)}_{\delta}(M) $$
such that $(M,g,u,\theta)^{(4)}$ is a stationary vacuum extension with Bartnik data $(\bartnik)$. 
\end{cor}
\begin{proof}
Choose $\mathcal{U}_{\varphi} = \mathcal{B}_{\varphi} (\mathcal{U})$, where $\mathcal{U}$ is from theorem \ref{main-thm}. Define $(g,u,\bar \theta)$ by 

\[(g,u,\bar \theta) := {\bf H}^{\mathcal{G}} \circ\mathcal{B}_{-\varphi} (\bartnik) \]

where ${\bf H}^{\mathcal{G}}$ is from theorem \ref{main-thm}. \\

Then the spacetime $(\mathcal{M}, \gsptime)$ defined by $(M, g,u,\bar \theta)^{(4)}$ is a stationary vacuum extension with Bartnik data $\mathcal{B}_{-\varphi} (\bartnik)$. In particular, the $\{t=0\}$ initial data set in this spacetime has Bartnik data $\mathcal{B}_{-\varphi} (\bartnik)$ on $\partial M$. We then need to show that there exists a unique initial data set in $(\mathcal{M}, \gsptime)$, defined by $t = f(x)$, with Bartnik data $(\bartnik)$ on $\partial M$ such that the corresponding $1$-form $\theta:= \bar \theta +df$ satisfies the $\theta$-geodesic gauge. If so, then $(M, g,u, \theta)^{(4)}$, with $\theta:= \bar \theta+ df$, is the unique stationary vacuum extension with Bartnik data $(\bartnik)$ as needed.

\vv

It then suffices to show that there exists a unique $\mathcal{C}^2$ function $f$, with $df = \mathcal{O}_1(r^{\delta})$, such that the rapidity of the $t=f(x)$ slice in $(\mathcal{M}, \gsptime)$ coincides with $\varphi$, and the $\theta$-geodesic gauge
\begin{equation}
\Delta_{g_{sc}} \bigg( (\bar \theta + df)(\dd{r}) \bigg) = 0
\end{equation}

holds. This is equivalent to the following elliptic boundary value problem for $f$: 

\begin{equation}
\begin{cases}
\Delta_{g_{sc}} \partial_r f = - \Delta_{g_{sc}} \left( \bar \theta(\dd{r}) \right), & \text{ on $M$}\\
{\bf n} (f) = \frac{\tanh \varphi}{N^{-1} + e^{2u} \theta({\bf n}) \tanh \varphi} & \text{on $\partial M$}\\
f = 0, & \text{on $\partial M$}\\
df = \mathcal{O}_1(r^{\delta})
\end{cases}
\end{equation}

where ${\bf n}$ is the unit normal of $\partial M$ in $(M, \gsp)$ and equation \ref{varphi f} was used. Note that due to the second boundary condition, tangential derivatives of $f$ vanish on $\partial M$ making the first boundary condition an equation on $\partial_r f$. The existence and uniqueness of a $\mathcal{C}^1$ function $\partial_r f$, and hence of $f$, follows similarly to the argument in the proof of proposition \ref{uniquef-prop}. 

\vv

The desired unique $\mathcal{C}^1$ map ${\bf H}^{\mathcal{G}}_{\varphi}$ is then given by: 
\[ (\bartnik) \mapsto (g,u, \bar \theta+df),\]
where $(g,u,\bar \theta) = {\bf H}^{\mathcal{G}} \circ \mathcal{B}_{-\varphi}(\bartnik)$ and $f$ is the solution to the above system.
\end{proof}

\section{Reduction of the problem} \label{reduction-sec}

First, we introduce some important notation that will be used for the rest of this paper (some of which has already been defined in earlier sections). \\

Let $(\mathcal{M}, \gsptime)$ be a stationary spacetime defined by $(M,g,u, \theta)^{(4)}$ (see definition \ref{stationary-def}), where $M = \R^3\setminus B_{r_0}$, $\mathcal{M} = \R\times M$, and $g$, $u$, $\theta$ are a Riemannian metric, a function, a 1-form on $M$ respectively. We define (or have defined) the following: 

\begin{itemize}
\item $\mathfrak{g}$ is the induced metric on the $\{t=0\}$ hypersurface, M, in $\mathcal{M}$. 
\item $\Pi$ is the second fundamental form of the $\{t=0\}$ hypersurface in $\mathcal{M}$. 
\item $\bar \theta := e^{2u} \theta$
\item $N$ is the function defined by $N:= \frac{\sqrt{1-e^{4u}|\theta|^2_g}}{e^u}$. 
\item $K_{\mathfrak{g}}$ is the second fundamental form of $\partial M$ in $(M, \gsp)$. 
\item $K$ is the second fundamental form of the sphere $S_r$ of radius $r$ in $(M, g)$. $\hat K$ is the traceless part of $K$. 
\item $\ngf$ is the inward unit normal on $\partial M$ in $(M, \gsp)$. 
\item $\ngg$ is the inward unit normal on $\partial M$ in $(M, g)$. 
\item $v_{\gsp}$ and $\lambda_{\gsp}$ are the vector field and the function on $S^2$ satisfying $\ngf = \lambda_{\gsp} \ngg+v_{\gsp}$
\item $\eta$ is the twist 1-form defined by $\eta := -\frac{1}{2} e^{4u} \star_g d\theta$.
\item The Bartnik boundary data in $\partial M$ of the initial data set $(M, \gsp, \Pi)$ is $(\bartnik)$, defined in definition \ref{bartnik-data-def}.
\item $\gamma_{\gsp}$ is the induced metric on $\partial M$ in $(M, \gsp)$. 
\item $g(r)$ is the induced metric on $S_r$ in $(M, g)$. 
\item $\slashed \nabla$ and $\cancel{div}$ is the connection and divergence on $(S^2, \gamma_{\mathbb{S}^2})$. 
\end{itemize}

We also remind the reader of the values of some key parameters for the Schwarzschild stationary vacuum extension $(M, g_{sc}, u_{sc} ,\theta_{sc})^{(4)}$. 

\begin{multicols}{2}
\begin{itemize}
\item $\mathfrak{g_{sc}} = \left( 1-\frac{2m_0}{r}\right)^{-1} dr^2 + r^2 \gamma_{\mathbb{S}^2}$
\item $g_{sc} = dr^2 + \left( 1-\frac{2m_0}{r} \right) r^2 \gamma_{\mathbb{S}^2} $
\item $trK_{\mathfrak{g_{sc}}} = \frac{2 \sqrt{1-\frac{2m_0}{r}}}{r} $
\item $ trK_{sc} = \frac{2(r-m_0)}{r(r-2m_0)}$
\item  $\hat K_{\gsc} = 0$
\item $\hat K_{sc} = 0$
\item $\gamma_{\mathfrak{g_{sc}}} = (r_0)^2 \gamma_{\mathbb{S}^2}$ 
\item $\gamma_{sc} = r_0(r_0-2m_0) \gamma_{\mathbb{S}^2}$ 
\item $\Pi_{\gsc} = 0$
\item $\theta_{\gsc} = 0$. 
\item $u_{sc} = \ln{\sqrt{1-\frac{2m_0}{r} }}$
\item ${R_{\partial M}}_{sc} = \frac{2}{r_0(r_0-2m_0)}$
\end{itemize}
\end{multicols}

\vv

Using a straightforward adaptation of the proof of the reduction theorem in \cite{ahmed}, we take advantage of the vanishing of the Weyl tensor in $3$ dimensions and utilize the geometry equations describing the evolution of the metric in the geodesic foliation to reduce the stationary vacuum equations to a coupled system of elliptic-transport equations.

\begin{redthm} \label{prop-reduction}
Let $(\bartnik)$ be Bartnik data. Let $(\mathcal{M}, \gsptime)$ be a stationary spacetime defined by $(M,g,u, \theta)^{(4)}$, where $g = dr^2 + g(r)$ ($g$ can be written globally in the geodesic gauge).  \\

$(M,g,u, \theta)^{(4)}$ is a stationary vacuum extension with Bartnik data $(\bartnik)$ if and only if 

\begin{align} \label{redeq1}
\Delta_g u +2e^{-4u} |\eta|^2_{g} &=0, &\quad \text{on $M$}  \\[0.3cm] \label{redeq2}
\partial_r trK + \frac{1}{2}trK^2 + |\hat K|^2 + 2(\partial_ru)^2 + 2e^{-4u} (\eta_r)^2 &= 0, &\quad \text{on $M$} \\[0.3cm] \label{redeq3}
 \begin{aligned}\begin{split} \nabla_r \hat K  + trK \hat K + \bigg[2 \slashed du \otimes \slashed du + 2e^{-4u} \eta^T \otimes \eta^T \\+ g(r) \left((\partial_r u)^2 + e^{-4u} (\eta_r)^2 - |\nabla u|^2 - e^{-4u} |\eta|^2 \right)\bigg]\end{split} \end{aligned} &= 0, &\quad \text{on $M$}\\[0.3cm] \label{redeq4}
d\eta &= 0, &\quad \text{on $M$}\\[0.3cm] 
\begin{aligned} \begin{split} 2|\slashed \nabla u|^2 + 2e^{-4u} |\eta|^2  - 2(\partial_r u)^2 \\ -2e^{-4u} (\eta_r)^2 - |\hat K|^2 - R_{\partial M} + \frac{1}{2}{trK}^2 \end{split} \end{aligned} &= 0, &\quad \text{on $\partial M$}\\[0.3cm] \label{redeq5}
2(\partial_r u) \slashed du + 2e^{-4u} \eta_r \eta^T - \cancel{div} (\hat K) + \frac{1}{2} \slashed dtrK  &=0, &\quad \text{on $\partial M$}\\[0.3cm] \label{redeq6}
 e^{-2u}g(r_0) - e^{2u} \theta^T \otimes \theta^T &= \gammao, &\quad \text{on $\partial M$}\\[0.3cm] \label{redeq7}
\mathrm{tr}K_{\mathfrak{g}} &= \trko, &\quad \text{on $\partial M$}\\[0.3cm]
 N \mathrm{div}_{\gamma_{\mathcal{B}}} (\bar {\theta^{\sharp}}^T) - N\Big(\mathrm{tr} K_{\mathcal{B}}\Big) \bar \theta(\ngf) &= \mathrm{tr}_{\partial M} \Pi_{\mathcal{B}}, &\quad \text{on $\partial M$} \\[0.3cm]
 \frac{N}{2} \mathcal{L}_{\bar \theta^{\sharp_{\gsp}}} \gsp (\ngf, \cdot) &= \omega_{\mathcal{B}} (\cdot), &\quad \text{on $\partial M$}
 \end{align}
\end{redthm}

\section{Proof of The Main Theorem} \label{proof-sec}

\subsection{The Setup} 

The space we will use for the Bartnik data $(\bartnik)$ is 
\begin{equation}
\mathcal{B}_{\mathcal{G}}^k := \mathcal{M}^{k+1}(S^2) \times H^k(S^2) \times H^k(S^2) \times \Omega^k(S^2)
\end{equation}

As for the Bartnik static extension problem, the contracted Codazzi equation in equation \eqref{redeq5} give rise to apparent obstructions to solvability that are in correspondence with the space of conformal Killing vector fields on $S^2$ (see section 5 in \cite{ahmed}). We will deal with these obstruction in the same way as in \cite{ahmed}: by introducing an artificial vector field $X$ to the definition of a solution. This means that the modified solution will now consist of a metric $g$, a function $u$, a 1-form $\theta$, {\it and} a vector field $X$. 

The definition of the artificial vector field $X$ and how it it enters the modified problem will be identical to what was done in \cite{ahmed}.  We will make the necessary definitions here and we refer the reader to section 5.1 and 5.2 in \cite{ahmed} for their motivation. 

\begin{defn} We define $\widehat{\mathcal{X}}^{2}_{\delta}(M)$ to be all vector fields $X \in \mathcal{X}^{2}_{\delta}(M)$ such that $X|_{\partial M}$ is tangent to $\partial M$ and $ \left( \mathcal{L}_{\dd{r}} X \right)^T = 0$ on $\partial M$.
\end{defn}

\begin{defn}
Define the space $\mathcal{X}_{\infty}$ as the space of conformal killing vector fields $X_{\infty}$ on $(M,g_{sc})$ of the form 
\begin{equation} \label{xinf}
X_{\infty} = f(r) \bigg( \cancel{div}_{\gamma_{\mathbb{S}^2}} (X_{CK}) \bigg) \dd{r} + h(r) \overline{X_{CK}}
\end{equation}
where $f = f(r)$ and $h= h(r)$ are smooth functions on $M$ such that $f= 0 $ and $h= 1$ on $\partial M$ and $X_{CK} $ is a conformal Killing vector field on $(\partial M, \gamma_{\mathbb{S}^2})$. \end{defn}

The artificial vector field $X$ will be chosen to live in the space $\bigg( \widehat{\mathcal{X}}^{2}_{\delta}(M) \oplus \mathcal{X}_{\infty}(M) \bigg)$. Given a vector field $X$, we will also defined the function $F(X)$ on $M$ by 
\begin{equation}
F(X) := e^{-r^4 \epsilon(r) |X|^2}
\end{equation}
where $|\cdot|$ is taken with respect to $g_{sc}$ and $\epsilon(r)$ is a smooth cut off function on $[r_0,\infty)$ satsifying $\epsilon(r) = 1$ for $r \geq r_0+2$ and $\epsilon(r) = 0$ for $r\leq r_0+1$. 

\vv

The space we will use for the modified solution $(g,u,X,\theta)$ is 

\begin{equation}
\mathcal{D}_{\mathcal{G}}^k :=  \mathcal{M}^k_{\delta}(M)  \times {\mathcal{A}^{(2,k+1)}_{\delta}}(M)  \times \bigg( \widehat{\mathcal{X}}^{2}_{\delta}(M) \oplus \mathcal{X}_{\infty}(M) \bigg) \times {\Omega_{\mathcal{G}}}^{(2,k+1)}_{\delta}(M)
\end{equation}

\vv

A subtle analytic issue arises in formulating the stationary vacuum extension problem as a nonlinear map between Banach manifolds: one must carefully specify the domain to ensure that the linearized operator is surjective. A critical component of the stationary vacuum equations is the vanishing of the exterior derivative of the twist 1-form, 

\begin{equation}
d \eta = 0
\end{equation}
where $\eta := -\frac{1}{2} e^{4u} \star_g d\theta$. 
 For $(g, u, \theta, X) \in \mathcal{D}{\mathcal{G}}^k$, the natural expectation would be that $d\eta$ lies in the space of exact 2-forms $d(\Omega^{(1,k)}_{\delta-1}(M))$. However, this space proves to be too large for our purposes. The key issue is that the $\theta$-geodesic gauge condition
\[\Delta_{sc} (\theta(\dd{r})) = 0\]
imposes constraints on $\theta$ that are not compatible with arbitrary exact 2-forms. Specifically, not every exact 2-form can arise as $d\eta$ for some $\theta$ satisfying our gauge condition. 

\vv

We must identify a proper subspace of exact 2-forms that captures precisely those arising from 1-forms $\theta$ in our gauge. This requires understanding how the gauge condition restricts the possible forms of $\eta$ and hence $d\eta$. The key insight is that the $\theta$-geodesic gauge condition implies special structure in the spherical harmonic decomposition of $\theta$. When we compute $\eta = -\frac{1}{2} e^{4u} \star_g d\theta$ and then $d\eta$, this structure is preserved in a specific way. After a detailed analysis (shown in the proof of the following lemma), we find the appropriate target space to be:

\begin{defn}
Define the space ${\Omega_{\mathcal{S}}}^{(1,k)}_{\delta-1}(M)$ to be the space of 1-forms $\sigma \in {\Omega}^{(1,k)}_{\delta-1}(M)$ of the form 

\begin{equation}
\sigma = \slashed d \alpha + \slashed \star \slashed d \partial_r \beta + \gamma dr
\end{equation}
where $\alpha \in \AHC{1}{k+1}{\delta}(M)$, $\beta \in \AHC{2}{k+2}{\delta+1}(M)$ satisfying $\beta(r_0) = 0$, and $\gamma \in \AHC{1}{k}{\delta-1}(M)$. 
\end{defn}

\begin{lem}
For any $g \in \mathcal{M}^k_{\delta}(M)$, $u \in \AHC{2}{k+1}{\delta}(M)$  and $\theta \in {\Omega_{\mathcal{G}}}^{(2,k+1)}_{\delta}(M)$, $$d\eta \in d \bigg({\Omega_{\mathcal{S}}}^{(1,k)}_{\delta-1}(M) \bigg)$$. 
\end{lem}
\begin{proof}
Let $g=dr^2 +g(r) \in \mathcal{M}^k_{\delta}(M)$, $u \in \AHC{2}{k+1}{\delta}(M)$ and $\theta \in {\Omega_{\mathcal{G}}}^{(2,k+1)}_{\delta}(M)$. We can then write $\theta$ in the form 

\begin{equation}
\theta =  \slashed d a +  \slashed \star \slashed d b + c dr
\end{equation}

where $a \in \AHC{2}{k+2}{\delta+1}(M)$, $b \in \AHC{2}{k+2}{\delta+1}(M)$, and $c \in \AHC{2}{k+1}{\delta}(M)$ satisfying the $\theta$-gauge: 
\begin{equation} \label{charmonic}
\Delta_{sc} c = 0, \quad \text{on $(M,g_{sc})$}
\end{equation}

Since $\eta = -\frac{1}{2} e^{4u} \star_g d\theta$, we have that $\eta \in {\Omega}^{(1,k)}_{\delta-1}(M)$  and, hence, can be written in the form: 
\begin{equation}
-2\eta = \slashed d \bar \alpha + \slashed \star \slashed d \bar \beta + \bar \gamma dr
\end{equation}

where $ \bar \alpha \in \AHC{1}{k+1}{\delta}(M),  \bar \beta \in  \AHC{1}{k+1}{\delta}(M), \bar \gamma \in  \AHC{1}{k}{\delta-1}(M)$.

\vv

We now introduce a piece of notation. Given a differential form $\omega$ on $M$, we denote by $\omega^{T_r}$ the projection of $\omega$ on the sphere $S_r$ and $\omega^{\perp_r}$ the restriction of $\omega$ on the normal bundle. In particular, we have that $\omega^{\perp_r} = \omega - \omega^{T_r}$. 

\vv

We compute $(d \eta)^{T_r}$  in terms of the functions $\bar \alpha, \bar \beta, \bar \gamma$ to be 

\begin{align}
-2 (d \eta)^{T_r} &= -2 \slashed d ( \eta^{T_r})\\
&= - \slashed \Delta \bar \beta d\sigma_{\mathbb{S}^2}
\end{align}

We now compute $(d \eta)^{T_r} $ in terms of $g, u, \theta$. First, we compute

\begin{align}
-2 \eta^{T_r} &= e^{4u} \star_g (d\theta)^{\perp_r}\\
& = e^{4u} \star_g ( dr \wedge \partial_r \slashed d a + dr \wedge \partial_r \slashed \star \slashed d b - dr \wedge \slashed d c)\\
&= e^{4u} (\partial_r \slashed \star \slashed d a + \partial_r \slashed d b - \slashed \star \slashed d c)\\
\end{align}

Taking the exterior derivative, we get
\begin{align}
-2 (d \eta)^{T_r}&=  -2 \slashed d (\eta^{T_r} )\\
&=e^{4u} \bigg( 4  \slashed du \wedge (\partial_r \slashed \star \slashed d a + \partial_r \slashed d b - \slashed \star \slashed d c)  \bigg)\\ 
& \qquad + e^{4u} \bigg( - \partial_r \slashed \Delta a + \slashed \Delta c  \bigg) d\sigma_{\mathbb{S}^2}
\end{align}

In light of the $\theta$-gauge in equation \eqref{charmonic}, we can write 

\begin{equation}
\slashed \Delta c = r(r-2m_0) (\partial^2_r c + trK_{sc} \partial_r c)
\end{equation}

and hence we deduce that 

\begin{align} \label{beta}
&- \slashed \Delta \bar \beta d\sigma_{\mathbb{S}^2} = e^{4u} \bigg( 4  \slashed du \wedge (\partial_r \slashed \star \slashed d a + \partial_r \slashed d b - \slashed \star \slashed d c)  \bigg)\\ 
& \qquad + e^{4u} \bigg( - \partial_r \slashed \Delta a + r(r-2m_0) (\partial^2_r c + trK_{sc} \partial_r c)   \bigg) d\sigma_{\mathbb{S}^2}\nonumber
\end{align}

We can rewrite the function multiplied to $d\sigma_{\mathbb{S}^2}$ in the second term on the right hand side as follows: 

\begin{align} 
e^{4u} \bigg( - \partial_r \slashed \Delta a + r(r-2m_0) (\partial^2_r c + trK_{sc} \partial_r c)   \bigg) & = \partial_r \left[ e^{4u} \bigg( -\slashed \Delta a + r(r-2m_0) (\partial_r c + trK_{sc} c \bigg) \right] \\
&\quad - \bigg(  \partial_r (e^{4u}) \slashed \Delta a + \partial_r (e^{4u} r(r-2m_0) ) \partial_r c + \partial_r(e^{4u}r(r-2m_0)) c \bigg)
\end{align}

Letting $(e_1,e_2)$ be an orthonormal frame on $(S^2, \gamma_{\mathbb{S}^2})$, we can finally observe the following after integrating in $r$: 

\begin{equation}
\int_{r_0}^r \left[ e^{4u} \bigg( 4 \slashed du \wedge (\partial_r \slashed \star \slashed d a + \partial_r \slashed d b - \slashed \star \slashed d c)  \bigg)(e_1,e_2) \right] dr \quad \text{lies in $\AHC{2}{k}{\delta+1}(M)$}
\end{equation}

\begin{equation}
\begin{aligned}
\left. \left[ e^{4u} \bigg( -\slashed \Delta a + r(r-2m_0) (\partial_r c + trK_{sc} c \bigg) \right] \right|^r_{r_0}\qquad \qquad \qquad  \\ + \int_{r_0}^r \left[ \bigg(  \partial_r (e^{4u}) \slashed \Delta a + \partial_r (e^{4u} r(r-2m_0) ) \partial_r c + \partial_r(e^{4u}r(r-2m_0)) c \bigg) \right] dr
\end{aligned}
\quad \text{lies in $\AHC{2}{k}{\delta+1}(M)$}
\end{equation}

Using equation \eqref{beta}, we deduce that 

\begin{equation}
\int_{r_0}^r \left[ \slashed \Delta \bar \beta \right] dr \quad \text{lies in $\AHC{2}{k}{\delta+1}(M)$}
\end{equation}

which implies that 

\begin{equation}
\beta:= \int_{r_0}^r  \bar \beta dr \quad \text{lies in $\AHC{2}{k+2}{\delta+1}(M)$}
\end{equation}

We then conclude that $\eta$ is of the form

\begin{equation}
-2\eta = \slashed d \bar \alpha + \slashed \star \slashed d \partial_r \beta + \bar \gamma dr
\end{equation}

where $\bar \alpha \in \AHC{1}{k+1}{\delta}(M)$, $\beta \in \AHC{2}{k+2}{\delta+1}(M)$ satisfying $\beta(r_0) = 0$, and $\bar \gamma \in \AHC{1}{k}{\delta-1}(M)$, implying that $d\eta \in d \bigg({\Omega_{\mathcal{S}}}^{(1,k)}_{\delta-1}(M) \bigg)$ as needed.

\end{proof}

\vv

We can now define $ \Phi^{\mathcal{G}}$ by:

\begin{multline*}
 \Phi^{\mathcal{G}}: \mathcal{B}_{\mathcal{G}}^k  \times \mathcal{D}_{\mathcal{G}}^k \to \\ 
 \qquad  \mathcal{A}^{(0,k-1)}_{\delta-2}(M) \times  L^2_{\delta-2} \bigg( [r_0,\infty); H^k(S^2) \bigg) \times  L^2_{\delta-2} \bigg( [r_0,\infty) ; \mathcal{H}^k (S^2) \bigg) \times \mathcal{X}^{0}_{\delta-2}(M) \times  d \bigg({\Omega_{\mathcal{S}}}^{(1,k)}_{\delta-1}(M) \bigg) \\
 \times H^{k-1}(S^2) \times \Omega^{k-1}(S^2) \times \mathcal{H}^{k}(S^2) \times  H^{k}(S^2) \times H^k(S^2) \times \Omega^k(S^2)
 \end{multline*}

\begin{align} \label{Gtilde}
& \Phi^{\mathcal{G}}(\bartnik , g, u, \theta, X ) := \begin{pmatrix} \Delta_g u +2e^{-4u} |\eta|^2_{g} &  \text{on $M$} \lv \\ 
\partial_r trK + \frac{1}{2}trK^2 + |\hat K|^2 + 2(\partial_ru)^2 +2e^{-4u} (\eta_r)^2 &  \text{on $M$} \lv \\
\begin{aligned}\begin{split} \nabla_r \hat K  + trK \hat K + \bigg[2 \slashed du \otimes \slashed du + 2e^{-4u} \eta^T \otimes \eta^T \\+ g(r) \left((\partial_r u)^2 + e^{-4u} (\eta_r)^2 - |\nabla u|^2 - e^{-4u} |\eta|^2 \right)\bigg]\end{split} \end{aligned} &  \text{on $M$} \lv \\
\Delta_{g, conf} (F(X) X) & \text{on $M$} \lv \\
d\eta & \text{on $M$} \lv \\
\begin{aligned} \begin{split} 2|\slashed \nabla u|^2 + 2e^{-4u} |\eta|^2  - 2(\partial_r u)^2 \\ -2e^{-4u} (\eta_r)^2 - |\hat K|^2 - R_{\partial M} + \frac{1}{2}{trK}^2 \end{split} \end{aligned} & \text{on $\partial M$} \lv \\
2(\partial_r u) \slashed du + 2e^{-4u} \eta_r \eta^T - \cancel{div} (\hat K) + \frac{1}{2} \slashed dtrK + \omega( g, X_{\infty})  &  \text{on $\partial M$}  \lv   \\
 e^{-2u} g(r_0) -e^{2u} \theta^T \otimes \theta^T - \gammao &  \text{on $ \partial M$} \lv \\
\mathrm{tr}K_{\mathfrak{g}} - \trko & \text{on $\partial M$} \lv \\
N \mathrm{div}_{\gamma_{\mathcal{B}}} (\bar {\theta^{\sharp}}^T) - N\Big(\mathrm{tr} K_{\mathcal{B}}\Big) \bar \theta(\ngf) - \mathrm{tr}_{\partial M} \Pi_{\mathcal{B}} & \text{on $\partial M$} \lv \\
  \frac{N}{2} \mathcal{L}_{\bar \theta^{\sharp_{\gsp}}} \gsp (\ngf, \cdot) - \omega_{\mathcal{B}} (\cdot) & \text{on $\partial M$}
  \end{pmatrix} 
\end{align}

\begin{defn}
We say that a 4-tuple $(g,u, X, \theta)$ is a modified solution with Bartnik data $(\bartnik)$ if $ \Phi^{\mathcal{G}}(\bartnik, g, u, X,\theta) = 0$.
\end{defn}

\begin{remark} \label{X=0 iff} In view of proposition \ref{prop-reduction}, a modified solution $(g,u,X,\theta)$ with Bartnik data $(\bartnik)$ and $X=0$ implies that $(M, g, u ,\theta)^{(4)}$ is a stationary vacuum extension with Bartnik data $(\bartnik)$. 
\end{remark}

The main tool to obtain the existence of the modified problem is the implicit function theorem on Banach manifolds (see \cite{implicit}), which is stated here for convenience. 

\begin{thm}
Let $U \subset E$, $V \subset F$ be open subsets of Banach spaces $E$ and $F$, and let $\Psi: U \times V \to G$ be a $C^r$ map to a Banach space $G$, with $r\geq 1$. For some $x_0 \in U$, $y_0 \in V$, assume the partial derivatives in the second argument $D_2 \Psi(x_0,y_0): F \to G$ is an isomorphism. Then there are neighbourhoods $U_0$ of $x_0$ and $W_0$ of $\Psi(x_0,y_0)$ and a unique $C^r$ map $H: U_0 \times W_0 \to V$ such that for all $(x,w) \in U_0 \times W_0$, $\Psi(x, H(x,w)) = w$. 
\end{thm}

 \vvv

The map $\Phi^{\mathcal{G}}$ is indeed $\mathcal{C}^1$ near $( \gammasc, \frac{1}{2} \trksc,0,0, g_{sc}, u_{sc}, 0 ,0)$ (see section 5.2 in \cite{ahmed}). We can then differentiate $\Phi^{\mathcal{G}}$ at $( \gammasc, \frac{1}{2} \trksc,0,0, g_{sc}, u_{sc}, 0,0 )$ and study its derivative.

\vv
Let $ D\Phi^{\mathcal{G}}_{sc}$ denote the derivative of $ \Phi^{\mathcal{G}}$ with respect to the last five components evaluated at $( \gammasc, \frac{1}{2} \trksc,0,0, g_{sc}, u_{sc}, 0,0 )$ where
\begin{multline*}
D \Phi^{\mathcal{G}}_{sc}: T_{g_{sc}}\mathcal{M}^k_{\delta} \times {\mathcal{A}^{(2,k+1)}_{\delta}} \times  \bigg( \widehat{\mathcal{X}}^{2}_{\delta}(M) \oplus \mathcal{X}_{\infty}(M) \bigg) \times {\Omega_{\mathcal{G}}}^{(2,k+1)}_{\delta}(M)\\
  \to  \mathcal{A}^{(0,k-1)}_{\delta-2}(M) \times  L^2_{\delta-2} \bigg( [r_0,\infty); H^k(S^2) \bigg) \times  L^2_{\delta-2}\bigg( [r_0,\infty) ; \mathcal{H}^k (S^2) \bigg) \times  d \bigg({\Omega_{\mathcal{S}}}^{(1,k)}_{\delta-1}(M) \bigg)\\
 \times H^{k-1}(S^2) \times \Omega^{k-1}(S^2) \times \mathcal{H}^{k}(S^2) \times  H^{k}(S^2)  \times  H^{k}(S^2) \times \Omega^k(S^2) 
 \end{multline*}

\begin{prop} \label{Dphi-iso}
 $D \Phi^{\mathcal{G}}_{sc}$ is an isomorphism.
\end{prop}
\begin{proof} 
The proof of this will be the content of section \eqref{linearized-sec} and \eqref{bvp-sec}. 
\end{proof}

\vv

We can now conclude the existence theorem for the modified problem.
\begin{thm}
There exists a neighbourhood $\mathcal{U}$ of $( \gammasc, \frac{1}{2} \trksc, 0, 0)$ in $\mathcal{B}_{\mathcal{G}}^k$ and a unique $\mathcal{C}^1$ map ${\bf H}^{\mathcal{G}}: (\bartnik) \mapsto (g, u, X,\theta)$ on $\mathcal{U}$ into $$ \mathcal{M}^k_{\delta}(M) \times \mathcal{A}^{(2,k+1)}_{\delta}(M) \times  \widehat{\mathcal{X}}^{2}_{\delta}(M) \oplus \mathcal{X}_{\infty}(M) \times  {\Omega_{\mathcal{G}}}^{(2,k+1)}_{\delta}(M)  $$ satisfying 
\begin{equation}  \Phi^{\mathcal{G}}\left(\bartnik, {\bf H}^{\mathcal{G}}\left(\bartnik \right)\right) = 0, \qquad \text{for all $(\bartnik) \in \mathcal{U}$}  \end{equation}

\end{thm}
\begin{proof}
Follows from proposition \eqref{Dphi-iso} and the implicit function theorem on Banach manifolds. 
\end{proof}

The vanishing of the artificial vector field $X$, given a solution $(g,u,X, \theta)$ to the modified problem, follows in the same way as in the Bartnik static extension problem (see section 5.3 in \cite{ahmed}). In particular, if $(g,u,X,\theta) \in \mathcal{D}_{\mathcal{G}}^k$ is a modified solution and the metric $\gamma_{\infty}$ on the sphere at infinity is close enough to the round metric in the $\mathcal{H}^k(S^2)$ norm, then $X = 0$ and $(M, g, u ,\theta)^{(4)}$ is a stationary vacuum extension. After possibly shrinking the neighbourhood $\mathcal{U}$ of $( \gammasc, \frac{1}{2} \trksc, 0, 0)$ and using the continuity of ${\bf H}^{\mathcal{G}}$, we finally conclude the main theorem as stated in section \ref{mainthm-sec}.

\subsection{The Linearized Problem} \label{linearized-sec}

Denote by $D\Phi^{\mathcal{G}}_{sc}$ the linearization of $\Phi^{\mathcal{G}}$ at $(\gammasc, \frac{1}{2} \trksc, 0,0, g_{sc}, u_{sc}, 0, 0)$. 

Let $\tilde g \in T_{g_{sc}} \mathcal{M}^k_{\delta}$, $\tilde u \in \mathcal{A}^{(2,k+1)}_{\delta}(M)$, $\tilde X \in \widehat{\mathcal{X}}^{2}_{\delta}(M)$, $\tilde \theta \in  {\Omega_{\mathcal{G}}}^{(2,k+1)}_{\delta}(M)$. For small $t$, let $g(t)$, $u(t)$, $X(t)$ and $\theta(t)$ be smooth 1-parameter families satisfying

\begin{multicols}{2}
\begin{itemize}
\item $g(0) = g_{sc}$
\item $u(0) = u_{sc}$
\item $X(0) = 0$
\item $\theta(0) = 0$
\item $g'(0) = \tilde g$
\item $u'(0)= \tilde u$
\item $X'(0) = \tilde X$
\item $\theta'(0) = \tilde \theta$
\end{itemize}
\end{multicols}

Define the following 

\begin{multicols}{1}
\begin{itemize}
\item $\widetilde{trK} := \ddt trK(t)$
\item $\widetilde{\hat K} := \ddt \hat K(t)$
\item $\tilde \gamma := \ddt g(t)(r_0)$
\item $\tilde \omega := \ddt \omega(g(t), X_{\infty}(t) )$
\end{itemize}
\end{multicols}

where $X_{\infty}(t)$ is the projection of $X(t)$ into the space $\mathcal{X}_{\infty}$. By definition of $\omega$, we have that $\tilde \omega$ is a conformal Killing field on $(S^2, g_{sc}(r_0))$.

\vv

We compute $D\Phi^{\mathcal{G}}_{sc}$ to be:

\begin{lem}
\begin{align} 
 D \Phi^{\mathcal{G}}_{sc} (\tilde g, \tilde u, \tilde X, \tilde \theta) &=  \left. \frac{d}{dt} \right|_{t=0} \Phi^{\mathcal{G}}(\gammasc, \frac{1}{2} \trksc, 0, 0, g(t), u(t), X(t), \theta(t) ) \nonumber  \\ 
 &=  \begin{pmatrix} \Delta_{g_{sc}} \tilde u +  (\partial_r u_{sc}) ( \widetilde{trK}) \lv \\ 
\partial_r \widetilde{trK} + trK_{sc} \widetilde{trK}  + 4(\partial_r u_{sc})( \partial_r \tilde u) \lv \\
\mathcal{L}_{\dd{r}} {\widetilde{\hat K}}   \lv \\
\Delta_{g_{sc}, conf} \tilde X  \lv \\ 
d\bigg(- \frac{1}{2} e^{4u_{sc}} \star_{g_{sc}} d \tilde \theta \bigg) \lv \\
 - 4(\partial_r u_{sc}) (\partial_r \tilde u) + trK_{sc} \left. \widetilde{trK} \right|_{\partial M}  +  \frac{4}{r_0(r_0-2m_0)} \tilde u +2 \slashed \Delta_{\gamma_{sc}} \tilde u \lv \\
2(\partial_r u_{sc}) \slashed d\tilde u - \cancel{div} (\widetilde{\hat K}) + \tilde \omega  \lv \\
 e^{-2u_{sc}} \tilde \gamma - 2 r_0^2 \, \tilde u \,  g_{S^2} \lv \\
e^{u_{sc}} \left( \left. \widetilde{trK} \right|_{\partial M} + \frac{2}{r_0} \tilde u - 2\partial_r \tilde u \right) \lv\\
 e^{-u_{sc}} \left( \cancel{div}_{\gamma_{\mathfrak{g_{sc}}}} (e^{2u_{sc}} \tilde \theta^T) - \Big(\mathrm{tr} K_{\gsc} \Big) e^{3u_{sc}} \tilde \theta_r \right) \lv \\
\frac{e^{-u_{sc}}}{2} \mathcal{L}_{e^{2u_{sc}} \tilde \theta^{\sharp_{\gsc}}} \gsc \left(e^{u_{sc}} \dd{r}, \cdot\right)
  \end{pmatrix}
 \end{align}
 
 \end{lem}
 \begin{proof}
We will only prove the linearized equation for the second Bartnik data. The rest follow by straightforward computations. 

We wish to compute 

\[ \ddt \mathrm{tr}K_{\mathfrak{g}} \]

We compute the second fundamental form $K_{\gsp}$ to be 

\begin{align}
K_{\gsp} &= \frac{1}{2} \mathcal{L}_{\ngf} \gsp\\
&= \frac{1}{2} \mathcal{L}_{\ngf} (e^{-2u} g) - \frac{1}{2} \mathcal{L}_{\ngf} (e^{2u} \theta \otimes \theta)\\
&= -\ngf(u) e^{-2u} g + \frac{1}{2} e^{-2u} \mathcal{L}_{\ngf} g - \frac{1}{2} \mathcal{L}_{\ngf} (e^{2u} \theta \otimes \theta)
\end{align}

Decomposing $\ngf = \lambda_{\gsp} \dd{r} + v_{\gsp}$ where $v_{\gsp}$ is tangent to $\partial M$, we deduce that for any vector $w$ tangent to $\partial M$, 

\begin{equation}
0 = \gsp(\ngf, w) = e^{-2u} g(v_{\gsp},w)+e^{2u} \theta\otimes \theta(\ngf,w)
\end{equation}

where we used the equation of $\gsp$ in terms of $g$, $u$, and $\theta$ (see equation \eqref{gsp-g})

Taking the derivative with respect to $t$, we get 
\begin{equation}
0 = e^{-2u_{sc}} g_{sc}(w,\tilde v_{\gsp})
\end{equation}
where $\tilde v = \ddt v_{\gsp}$, implying that $\tilde v_{\gsp} = 0$. A similar calculation shows that 
\begin{equation}
\tilde \lambda_{\gsp} := \ddt \lambda_{\gsp} = \tilde u e^{u_{sc}}
\end{equation}

Taking the $\gsp$-trace of $K_{\gsp}$ and taking the derivative with respect to $t$, we deduce that 

\begin{equation}
\ddt \mathrm{tr}K_{\gsp} = e^{u_{sc}} \left( \left. \widetilde{trK} \right|_{\partial M} + \frac{2}{r_0} \tilde u - 2\partial_r \tilde u \right)
\end{equation}
 \end{proof}

Note that the equations for $(\tilde g, \tilde u, \tilde X)$ decouple completely from the equations for $\tilde \theta$, which are the $5^{th}$, $10^{th}$ and $11^{th}$ equations. More specifically, we have that 

\begin{align} 
 D \Phi^{\mathcal{G}}_{sc} (\tilde g, \tilde u, \tilde X, \tilde \theta)&=  \begin{pmatrix} \Delta_{g_{sc}} \tilde u +  (\partial_r u_{sc}) ( \widetilde{trK}) \lv \\ 
\partial_r \widetilde{trK} + trK_{sc} \widetilde{trK}  + 4(\partial_r u_{sc})( \partial_r \tilde u) \lv \\
\mathcal{L}_{\dd{r}} {\widetilde{\hat K}}   \lv \\
\Delta_{g_{sc}, conf} \tilde X  \lv \\ 
0 \lv \\
 - 4(\partial_r u_{sc}) (\partial_r \tilde u) + trK_{sc} \left. \widetilde{trK} \right|_{\partial M}  +  \frac{4}{r_0(r_0-2m_0)} \tilde u +2 \slashed \Delta_{\gamma_{sc}} \tilde u \lv \\
2(\partial_r u_{sc}) \slashed d\tilde u - \cancel{div} (\widetilde{\hat K}) + \tilde \omega  \lv \\
 e^{-2u_{sc}} \tilde \gamma - 2 r_0^2 \, \tilde u \,  g_{S^2} \lv \\
e^{u_{sc}} \left( \left. \widetilde{trK} \right|_{\partial M} + \frac{2}{r_0} \tilde u - 2\partial_r \tilde u \right) \lv\\
 0  \lv \\
0
  \end{pmatrix} \\
  &+  \begin{pmatrix}0 \lv \\ 
0 \lv \\
0 \lv \\
0  \lv \\ 
d\bigg(- \frac{1}{2} e^{4u_{sc}} \star_{g_{sc}} d \tilde \theta \bigg) \lv \\
 0 \lv \\
0 \lv \\
0 \lv \\
0 \lv\\
 e^{-u_{sc}} \left( \cancel{div}_{\gamma_{\mathfrak{g_{sc}}}} (e^{2u_{sc}} \tilde \theta^T) - \Big(\mathrm{tr} K_{\gsc} \Big) e^{3u_{sc}} \tilde \theta_r \right) \lv \\
\frac{e^{-u_{sc}}}{2} \mathcal{L}_{e^{2u_{sc}} \tilde \theta^{\sharp_{\gsc}}} \gsc \left(e^{u_{sc}} \dd{r}, \cdot\right)
  \end{pmatrix}
 \end{align}

 Note also that the equations for $(\tilde g, \tilde u, \tilde X)$ in the first term in the right hand side of the above equation are identical to the equations in the definition of the operator $D\Phi_{sc}$ in equation 5.81 in \cite{ahmed}, which is an isomorphism according to proposition 5.9 in \cite{ahmed}. To show that $D\Phi^{\mathcal{G}}_{sc}$ is an isomorphism, it suffices to prove the following proposition:
\begin{prop} \label{psi-prop}
Define the operator $\Psi$ by 
\[\Psi: {\Omega_{\mathcal{G}}}^{(2,k+1)}_{\delta}(M) \to  d \bigg({\Omega_{\mathcal{S}}}^{(1,k)}_{\delta-1}(M) \bigg)  \times  H^{k}(S^2) \times \Omega^k(S^2)  \]

\begin{equation}
\Psi(\tilde \theta ) := \begin{pmatrix}d\bigg(- \frac{1}{2} e^{4u_{sc}} \star_{g_{sc}} d \tilde \theta \bigg) \lv \\
 e^{-u_{sc}} \left( \cancel{div}_{\gamma_{\gsc}} (e^{2u_{sc}} \tilde \theta^T) - \Big(\mathrm{tr} K_{\gsc} \Big) e^{3u_{sc}} \tilde \theta_r\right)  \lv \\
\frac{e^{-u_{sc}}}{2} \mathcal{L}_{e^{2u_{sc}} \tilde \theta^{\sharp_{\gsc}}} \gsc \left(e^{u_{sc}} \dd{r}, \cdot\right)
  \end{pmatrix}
\end{equation}

$\Psi$ is an isomorphism. 
\end{prop}
\begin{proof}
The proof of this will be the content of the next section. 
\end{proof}

\subsection{The BVP for the 1-form $\tilde \theta$} \label{bvp-sec}

In this section, we prove proposition \ref{psi-prop}, which reduces to proving the wellposedness of the following BVP for the 1-form $\tilde \theta$. 

\begin{thm} \label{bvp-thm}

Let $\sigma \in {\Omega_{\mathcal{S}}}^{(1,k)}_{\delta-1}(M)$, $h \in H^k(S^2)$ and $\Lambda \in \Omega^{k}(S^2)$. There exists a unique $\tilde \theta \in {\Omega_{\mathcal{G}}}^{(2,k+1)}_{\delta}(M)$ satisfying 

\begin{equation} \label{bvp}
\begin{cases}
4 \partial_r u_{sc}  \left(  dr \wedge \sgsc d\tilde \theta \right) + d\sgsc d\tilde \theta = d\sigma, & \text{in $M$}\\
 \cancel{div} (\tilde \theta^T) - 2(r_0-2m_0) \tilde \theta_r = h, & \text{on $\partial M$}\\
\left( \mathcal{L}_{\dd{r}} \tilde \theta \right)^T +\slashed d \tilde \theta_r - \frac{2(r_0-3m_0)}{r_0(r_0-2m_0)} \tilde \theta^T = \Lambda,& \text{on $\partial M$}
\end{cases}
\end{equation}
\end{thm}

\vv

 The second boundary condition comes from the following computation: in spherical coordinates $(\phi^1, \phi^2)$,  
 
  \begin{align}
 \frac{e^{-u_{sc}}}{2} \mathcal{L}_{e^{2u_{sc}} \tilde \theta^{\sharp_{\gsc}}} \gsc \left(e^{u_{sc}} \dd{r}, \dd{\phi^i} \right) &= \frac{1}{2} \left( \nabla^{\gsc}_0 (e^{2u_{sc}} \tilde \theta)_i + \nabla^{\gsc}_i (e^{2u_{sc}} \tilde \theta)_0 \right)\\
 &=  \frac{e^{2u_{sc}}}{2}  \left( 2 \partial_r u_{sc} \tilde \theta_i + \partial_r \tilde \theta_i + \partial_i \tilde \theta_r - 2{\Gamma}_{0i}^{k}\tilde \theta_k \right)\\
 &=  \frac{e^{2u_{sc}}}{2} \left( \left( \mathcal{L}_{\dd{r}} \tilde \theta \right)^T +\slashed d \tilde \theta_r - \frac{2(r_0-3m_0)}{r_0(r_0-2m_0)} \tilde \theta^T \right)_i
 \end{align}
 The content of this section is the proof of this theorem.

 \vv

We utilize the spherical symmetry of $(M,g_{sc})$ to reduce the BVP into ODEs on the spherical harmonic coefficients. For $\ell \in \Z_{\geq0}$ and $-\ell\leq m \leq \ell$, let $Y_{m\ell}$ be the spherical harmonics on $S^2$ normalized with respect to the unit sphere $(S^2, \gamma_{\mathbb{S}^2})$. In particular, the following holds
\begin{equation}
\slashed \Delta Y_{m\ell} = -\ell(\ell+1) Y_{m\ell}
\end{equation}
 
 We also define the following two 1-forms on $(S^2, \gamma_{\mathbb{S}^2})$: 
\[
\slashed d Y_{m\ell} \quad \text{and} \quad \slashed \star \slashed dY_{m\ell}
\]
 We extend $Y_{m\ell}$ and the above two 1-forms to $M$ so that they are constant in $r$. We will also extend the volume form $\dsig$ on $(S^2, \gamma_{\mathbb{S}^2})$ in the same manner. For simplicity, we will use the same notation for the extensions; it should be clear from context which on we are referring to. In particular, the extensions $Y_{m\ell}$, $\dy$, $\sdy$ and $\dsig$ satisfy the following on $M$
 \[\mathcal{L}_{\dd{r}} Y_{m\ell} = 0, \quad \mathcal{L}_{\dd{r}} \dy = 0, \quad \mathcal{L}_{\dd{r}} \sdy = 0, \quad \mathcal{L}_{\dd{r}} \dsig = 0   \]

We have the following spherical harmonic decomposition theorem for functions in  $\AHC{t}{k}{\delta}$ spaces (see section 3 in \cite{ahmed}). 
\begin{prop} \label{sph-prop}  Every function $f \in \AHC{0}{0}{\delta}(M)$ admits a unique decomposition 
\begin{equation} \label{sph-f}
f = \summ f_{m\ell}(r) Y_{m\ell}
\end{equation}
where $f_{m\ell} \in L^2([r_0,\infty) \cap C^0([r_0,\infty))$. Furthermore, for $k\geq t\geq 0$,  $f \in \AHC{t}{k}{\delta}(M)$ if and only if
\begin{equation} \summ \sum_{t'=0}^t [1+\ell(\ell+1)]^{k-t'} (\normmmH{f_{m\ell}^{(t')}}{0}{\delta-t'}^2+\normmmC{f_{m\ell}^{(t')}}{0}{\delta-t'}^2) < \infty \end{equation}
In particular, the square root of the left side is an equivalent norm on $\AHC{t}{k}{\delta}(M)$. 

\end{prop}

We decompose $\tilde \theta \in {\Omega_{\mathcal{G}}}^{(2,k+1)}_{\delta}(M) $, $\sigma \in {\Omega_{\mathcal{S}}}^{(1,k)}_{\delta-1}(M)$, $h \in H^k(S^2)$ and $\Lambda \in \Omega^{k}(S^2)$ (see proposition \ref{hodge-prop}): 

\begin{itemize}

\item \begin{equation}
\tilde \theta = \slashed d a + \slashed \star \slashed d b + c dr
\end{equation}
where 
\begin{equation} \label{sph-theta}
a = \summm a_{m\ell}(r)Y_{m\ell}, \quad b = \summm \sqrt{r(r-2m_0)}b_{m\ell}(r)Y_{m\ell}, \quad c =  \summ c_{m\ell}(r) Y_{m\ell}
\end{equation}

with norm: 

\begin{equation}
\norm{\tilde \theta}_{\Omega^{(2,k+1)}_{\delta}}^2 = \norm{a}_{\AHC{2}{k+2}{\delta+1}}^2 + \norm{b}_{\AHC{2}{k+2}{\delta+1}}^2 + \norm{c}_{\AHC{2}{k+1}{\delta}}^2
\end{equation}

\item \begin{equation}
\sigma = \slashed d \alpha + \slashed \star \slashed d \partial_r\beta + \zeta dr
\end{equation}

where 
\begin{equation} \label{sph-sigma}
\alpha= \summm \sqrt{r(r-2m_0)} \alpha_{m\ell}(r)Y_{m\ell}, \quad \beta = \summm \beta_{m\ell}(r)Y_{m\ell}, \quad \zeta =  \summ \sqrt{r(r-2m_0)}\zeta_{m\ell}(r) Y_{m\ell}
\end{equation}

with norm 

\begin{equation}
\norm{ \sigma}_{\Omega^{(1,k)}_{\delta-1}}^2 = \norm{\alpha}_{\AHC{1}{k+1}{\delta}}^2 + \norm{\beta}_{\AHC{2}{k+2}{\delta+1}}^2 + \norm{\zeta}_{\AHC{1}{k}{\delta-1}}^2
\end{equation}

\item 
\begin{equation} \label{sph-h}
 h = \summ h_{m\ell} Y_{m\ell}
\end{equation}

with norm 

\begin{equation}
\norm{h}_{H^k(S^2)}^2 = \summ [1+\ell(\ell+1)]^k |h_{m\ell}|^2
\end{equation}
\item 
\begin{equation} \label{sph-lambda}
 \Lambda = \summm \lambda_{m\ell} \dy + \xi_{m\ell} \sdy 
\end{equation}

with norm 

\begin{equation}
\norm{\Lambda}_{\Omega^k(S^2)}^2 =  \summm [1+\ell(\ell+1)]^{k+1} \Big(|\lambda_{m\ell}|^2+|\xi_{m\ell}|^2\Big)
\end{equation}
 \end{itemize}
 
 \vv
 
 Furthermore, the $\theta$-gauge satisfied by $\tilde \theta$, 

\begin{equation}
\Delta_{g_{sc}} \left( \tilde \theta(\dd{r})\right) = 0\end{equation}

is equivalent to the following ODE on the functions $c_{m\ell}$ on $[r_0,\infty)$:
\begin{equation}
r(r-2m_0)c_{m\ell}''(r) + 2(r-m_0)c_{m\ell}'(r) - \ell(\ell+1)c_{m\ell}(r) = 0
\end{equation}

\vv

The BVP in equation \eqref{bvp} reduces to a system on the spherical harmonic coefficients according to the following lemma. 
\begin{lem}
Let $\tilde \theta \in {\Omega_{\mathcal{G}}}^{(2,k+1)}_{\delta}(M) $, $\sigma \in {\Omega_{\mathcal{S}}}^{(1,k)}_{\delta-1}(M)$, $h \in H^k(S^2)$ and $\Lambda \in \Omega^{k}(S^2)$ with spherical harmonic decomposition as in equations \eqref{sph-theta}, \eqref{sph-sigma}, \eqref{sph-h}, and \eqref{sph-lambda}. 

The BVP in equation \eqref{bvp} holds if and only if the following system of ODEs on the spherical harmonic coefficients are satisfied:

\begin{itemize}

\item \textit{For $\ell = 0$:}

\begin{align}\label{sys1}
\begin{cases} 
r(r-2m_0)c_{00}''(r) + 2(r-m_0)c_{00}'(r) = 0, \quad \text{on $[r_0,\infty)$} \lv\\
-2(r_0-2m_0) c_{00}(r_0) = h_{00}
\end{cases}
\end{align}

\item  \textit{For $\ell \geq 1$:}
\begin{align} \label{sys2}
\begin{cases}
r(r-2m_0)b_{m\ell}''(r) + 2(r-m_0) b_{m\ell}'(r) - \ell(\ell+1) b_{m\ell}(r) + 4m_0 b_{m\ell}'(r) + \frac{m_0(4r-5m_0)}{r(r-2m_0)} b_{m\ell}(r)\\ =  r(r-2m_0)\alpha_{m\ell}'(r) + (r-m_0)\alpha_{m\ell}(r) - r(r-2m_0) \zeta_{m\ell}(r), \qquad \text{on $[r_0,\infty)$} \lv\\
b_{m\ell}'(r_0) - \frac{(r_0-5m_0)}{r_0(r_0-2m_0)} b_{m\ell}(r_0) = \frac{\xi_{m\ell}}{\sqrt{r_0(r_0-2m_0)}}
\end{cases}
\end{align}

\begin{equation} \label{sys3}
\begin{cases}
a'_{m\ell}(r) - c_{m\ell}(r) = \beta'_{m\ell}(r), & \text{on $[r_0,\infty)$}\lv \\
r(r-2m_0)c_{m\ell}''(r) + 2(r-m_0)c_{m\ell}'(r) - \ell(\ell+1)c_{m\ell}(r) = 0, & \text{on $[r_0,\infty)$} \lv\\
a_{m\ell}'(r_0) + c_{m\ell}(r_0) - \frac{2(r_0-3m_0)}{r_0(r_0-2m_0)} a_{m\ell}(r_0) = \lambda_{m\ell}\lv \\
-a_{m\ell}(r_0) \ell(\ell+1) - 2(r_0-2m_0)c_{m\ell}(r_0) = h_{m\ell}
\end{cases}
\end{equation}

\end{itemize}
\end{lem}
\begin{proof}
The lemma follows by a straightforward computation that uses the following identities: 
\begin{itemize}
\item $d\dy = 0$
\item $d\sdy = -\ell(\ell+1) Y_{m\ell} d\sigma_{\mathbb{S}^2}$.
\item $\sgsc (dr \wedge \dy) = \sdy$
\item $\sgsc (dr \wedge \sdy) = -\dy$
\item $\sgsc\dsig = \frac{1}{r(r-2m_0)} dr$
\end{itemize}

\end{proof}

We can now formulate theorem \ref{bvp-thm} in terms of the spherical harmonic decomposition: 

\begin{thm}
Let $\alpha \in \AHC{1}{k+1}{\delta}(M)$, $\beta \in \AHC{2}{k+2}{\delta+1}(M)$ satisfying $\beta(r_0) = 0$, $\zeta \in \AHC{1}{k}{\delta-1}(M)$, $h \in H^k(S^2)$, and $\Lambda \in \Omega^k(S^2)$ with spherical harmonic decomposition as in equations \eqref{sph-sigma}, \eqref{sph-h}, and \eqref{sph-lambda}. There exists a unique $a,b \in \AHC{2}{k+2}{\delta+1}(M)$ and $c \in \AHC{2}{k+1}{\delta}(M)$ with spherical harmonic decomposition as in equation\eqref{sph-theta} satisfying the system in equations \eqref{sys1}, \eqref{sys2}, and \eqref{sys3}.  
\end{thm}

The rest of this section is devoted to proving this theorem. We fix an $\alpha, \beta, \zeta, h,$ and $\Lambda$ as in the statement of the theorem. The proof will proceed in the following steps. 

\begin{enumerate}[label=\textbf{Step \arabic*: }]
\item  We will show that for each $\ell,m$, where $\ell\geq 1$, there exists a unique $b_{m\ell} \in H^2_{\delta}([r_0,\infty)) \cap C^2_{\delta}[r_0,\infty)$ satisfying the system in \eqref{sys2}. Then we will show that the function $b$ on $M$, with spherical harmonic decomposition as in equation \eqref{sph-theta}, indeed lives in $\AHC{2}{k+2}{\delta+1}(M)$.  

\item We will solve for $a'_{m\ell}(r_0)$, $a_{m\ell}(r_0)$, and $c_{m\ell}(r_0)$, and estimate them by $ \beta, \Lambda$, and $h$.

\item We will show that for each $\ell,m$, where $\ell\geq 0$, there exists a unique $c_{m\ell} \in H^2_{\delta}([r_0,\infty)) \cap C^2_{\delta}[r_0,\infty)$ satisfying the system in \eqref{sys1} and \eqref{sys3}. Then we will show that the function $c$ on $M$, with spherical harmonic decomposition as in equation \eqref{sph-theta}, indeed lives in $\AHC{2}{k+1}{\delta}(M)$. 

\item  We will show that for each $\ell,m$, where $\ell\geq 1$, there exists a unique $a_{m\ell} \in H^2_{\delta+1}([r_0,\infty)) \cap C^2_{\delta}[r_0,\infty)$ satisfying the system in \eqref{sys3}. Then we will show that the function $a$ on $M$, with spherical harmonic decomposition as in equation \eqref{sph-theta}, indeed lives in $\AHC{2}{k+2}{\delta+1}(M)$. 
\end{enumerate}
 
 \vv
 
We first introduce a piece of notation. Given two quantities $Y, Z$, we will write $Y \lesssim Z$ if there exists a constant $C>0$ depending only on $r_0$, $m_0$ and $\delta$, such that $Y \leq C Z$. 

\vv

Before we begin proving the above steps, we will state a well-posedness result for a particular second order ODE that is relevant to the system we are solving in equations \eqref{sys1} - \eqref{sys3}.  

\begin{prop} \label{ode-prop}
For $\ell \geq 1$, $-\ell\leq m\leq \ell$, define the operators $\Phi_{m\ell}^{\mathcal{D}}$ and $\Phi^{\mathcal{N}}_{m\ell}$ by the following: for a function $f$ on $[r_0,\infty)$,  

\begin{equation} 
\Phi_{m\ell}^{\mathcal{D}} (f) = \begin{pmatrix} f''(r) + \frac{2(r-m_0)}{r(r-2m_0)} f'(r) - \frac{\ell(\ell+1)}{r(r-2m_0)} f(r)  \lv\\  f(r_0)  \end{pmatrix}
\end{equation}

\begin{equation} 
\Phi_{m\ell}^{\mathcal{N}} (f) = \begin{pmatrix} f''(r) + \frac{2(r-m_0)}{r(r-2m_0)} f'(r) - \frac{\ell(\ell+1)}{r(r-2m_0)} f(r)  \lv\\ f'(r_0)  \end{pmatrix}
\end{equation}

For $\delta \in (-1, -\frac{1}{2})$, 

\[ \Phi_{m\ell}^{\mathcal{D}} : H^2_{\delta}([r_0,\infty)) \cap C^2_{\delta}[r_0,\infty) \to L^2_{\delta-2}([r_0,\infty)) \cap C^0_{\delta-2}([r_0,\infty) \times \R \quad \text{is an isomorphism}\]
\[ \Phi_{m\ell}^{\mathcal{N}} : H^2_{\delta}([r_0,\infty)) \cap C^2_{\delta}[r_0,\infty) \to L^2_{\delta-2}([r_0,\infty)) \cap C^0_{\delta-2}([r_0,\infty) \times \R \quad \text{is an isomorphism}\]

\vv

Furthermore, the following estimates, uniform in $m$ and $\ell$, hold:\\
 for any $f \in H^2_{\delta}([r_0,\infty)) \cap C^2_{\delta}[r_0,\infty)$, 

\begin{equation} 
\begin{aligned} \label{estimate-f H D}
\normmmH{f''}{0}{\delta-2}^2 + \left[ 1+\ell(\ell+1)\right] \normmmH{f'}{0}{\delta-1}^2 + \left[ 1+\ell(\ell+1)\right]^2 \normmmH{f}{0}{\delta}^2 \\ \lesssim \normmmH{\pi_1 \circ \Phi^{\mathcal{D}}_{m\ell}(f)}{0}{\delta-2}^2 + \left[ 1+\ell(\ell+1)\right]^{3/2}|\pi_2 \circ \Phi^{\mathcal{D}}_{m\ell}(f)|^2
\end{aligned}
\end{equation}

\begin{equation} 
\begin{aligned} \label{estimate-f C D}
\normmmC{f''}{0}{\delta-2}^2 + \left[ 1+\ell(\ell+1)\right] \normmmC{f'}{0}{\delta-1}^2 + \left[ 1+\ell(\ell+1)\right]^2 \normmmC{f}{0}{\delta}^2 \\  \lesssim \normmmC{\pi_1 \circ \Phi^{\mathcal{D}}_{m\ell}(f)}{0}{\delta-2}^2 + \left[ 1+\ell(\ell+1)\right]^2|\pi_2 \circ \Phi^{\mathcal{D}}_{m\ell}(f)|^2
\end{aligned}
\end{equation}

\begin{equation} 
\begin{aligned} \label{estimate-f H N}
\normmmH{f''}{0}{\delta-2}^2 + \left[ 1+\ell(\ell+1)\right] \normmmH{f'}{0}{\delta-1}^2 + \left[ 1+\ell(\ell+1)\right]^2 \normmmH{f}{0}{\delta}^2 \\ \lesssim \normmmH{\pi_1 \circ \Phi^{\mathcal{N}}_{m\ell}(f)}{0}{\delta-2}^2 + \left[ 1+\ell(\ell+1)\right]^{1/2}|\pi_2 \circ \Phi^{\mathcal{N}}_{m\ell}(f)|^2
\end{aligned}
\end{equation}

\begin{equation} 
\begin{aligned} \label{estimate-f C N}
\normmmC{f''}{0}{\delta-2}^2 + \left[ 1+\ell(\ell+1)\right] \normmmC{f'}{0}{\delta-1}^2 + \left[ 1+\ell(\ell+1)\right]^2 \normmmC{f}{0}{\delta}^2 \\ \lesssim  \normmmC{\pi_1 \circ \Phi^{\mathcal{N}}_{m\ell}(f)}{0}{\delta-2}^2 + \left[ 1+\ell(\ell+1)\right]|\pi_2 \circ \Phi^{\mathcal{N}}_{m\ell}(f)|^2 
\end{aligned}
\end{equation}

where $\pi_1$, $\pi_2$ are the projection maps into the first and second component respectively. 

\end{prop}
\begin{proof}
Follows from a straightforward adaptation of the proof in section 3 of \cite{ahmed}. 
\end{proof}

\section*{Step 1: Solving for $b$}

We will show that for each $\ell,m$, where $\ell\geq 1$, there exists a unique $b_{m\ell} \in H^2_{\delta}([r_0,\infty)) \cap C^2_{\delta}[r_0,\infty)$ satisfying the system in \eqref{sys1}. 

\vv

Motivated by the system in \eqref{sys2} satisfied by $b_{m\ell}$, we define the operator $\Phi^{\mathcal{G}}_{m\ell}$, for $\ell \geq 1$, $-\ell\leq m\leq \ell$, from $H^2_{\delta}([r_0,\infty)) \cap C^2_{\delta}[r_0,\infty)$ to $L^2_{\delta-2}([r_0,\infty)) \cap C^0_{\delta-2}([r_0,\infty) \times \R$ by

\begin{equation}
\Phi^{\mathcal{G}}_{m\ell} (f) = \Phi^{\mathcal{N}}_{m\ell}(f) + \begin{pmatrix}  \frac{4m_0}{r(r-2m_0)} f'(r) + \frac{m_0(4r-5m_0)}{r^2(r-2m_0)^2}f(r)  \lv\\ -\frac{(r_0-5m_0)}{r_0(r_0-2m_0)} f(r_0)  \end{pmatrix}
\end{equation}

\vv

Due to standard compact embedding theorems of the spaces $H^k_{\delta}([r_0,\infty))$ and $C^k_{\delta}([r_0,\infty))$, we note that $\Phi_{m\ell}^{\mathcal{G}}$ is a compact perturbation of $\Phi^{\mathcal{N}}_{m\ell}$, which an isomorphism due to proposition \ref{ode-prop}. Hence, $\Phi_{m\ell}^{\mathcal{G}}$ is Fredholm of index $0$. To show that $\Phi_{m\ell}^{\mathcal{G}}$ is an isomorphism, it then suffices to show that its kernel is trivial, which is shown in the next proposition. 

\begin{prop}
Let $\ell \in \Z_{>0}$. Let $f_{\ell} \in H^2_{\delta}([r_0,\infty)) \cap C^2_{\delta}[r_0,\infty)$ satisfy the following system 
 
\begin{align} \label{sys2f}
\begin{cases}
r(r-2m_0)f_{\ell}''(r) + 2(r-m_0) f_{\ell}'(r) - \ell(\ell+1) f_{\ell}(r)\\
 =- 4m_0 f_{\ell}'(r) - \frac{m_0(4r-5m_0) }{r(r-2m_0)}f_{\ell}(r), \qquad \text{on $[r_0,\infty)$} \lv\\
f_{\ell}'(r_0) - \frac{(r_0-5m_0)}{r_0(r_0-2m_0)} f_{\ell}(r_0) = 0
\end{cases}
\end{align}

Then $f=0$. 
\end{prop}
\begin{proof}

For $\ell \in\Z_{>0}$, define $F_{\ell}$ as the function on $[r_0,\infty)$ defined by 

\[F_{\ell} (r) = \sqrt{r(r-2m_0)} f_{\ell}(r)\]

Letting $C_{\ell}:= F_{\ell}(r_0)$, we find that $F_{\ell}$ satisfies: 

\begin{align} \label{sys2F}
\begin{cases}
\left[ \frac{(r-2m_0)^2}{r^2} F_{\ell}'(r)\right]' = \ell(\ell+1) \frac{r-2m_0}{r^3} F_{\ell}(r), \qquad \text{on $[r_0,\infty)$} \lv\\
F_{\ell}'(r_0) = \frac{2(r_0-3m_0)}{r_0(r_0-2m_0)} C_{\ell} \\
F_{\ell}(r_0) = C_{\ell}
\end{cases}
\end{align}

If $C_{\ell} = 0$, then $F_{\ell} = 0$ implying that $f_{\ell} = 0$ as needed. We will assume that $C_{\ell} \neq 0$ and show that $r^{-1}F_{\ell}$ is unbounded, which leads to a contradiction as we know that $\lim_{r\to \infty} f_{\ell}(r) = 0$. 

\vv

By dividing by $C_{\ell}$ and replacing $F_{\ell}/C_{\ell}$ with $F_{\ell}$, we can assume without loss of generality that $C_{\ell} = 1$. We will prove that $F_{\ell}(r) > 0 $ for $r \in [r_0,\infty)$ and $\lim_{r\to \infty} r^{-1}F_{\ell}(r) = \infty$ by induction on $\ell$.

\vv

For the case $\ell = 1$, we find that the unique solution for the above initial value problem is

\begin{equation}
F_1(r) = \frac{r_0-2m_0}{r_0^3} \frac{r^3}{r-2m_0}
\end{equation}

which satisfies $F_{1}(r) > 0 $ for $r \in [r_0,\infty)$  and $\lim_{r\to \infty} r^{-1}F_1(r) = \infty$ as needed. 

\vv

Now suppose that $F_{\ell}(r) > 0 $ for $r \in [r_0,\infty)$ and $\lim_{r\to \infty} r^{-1}F_{\ell} (r) = \infty$ for some $\ell \geq 1$. We define the function $H_{\ell}:= F_{\ell+1} - F_{\ell}$ on $[r_0,\infty)$, which satisfies the following system: 

\begin{align} \label{sys2H}
\begin{cases}
\left[ \frac{(r-2m_0)^2}{r^2} H_{\ell}'(r)\right]' = (\ell+1) \frac{r-2m_0}{r^3} \Big( (\ell+2)H_{\ell}(r) + 2F_{\ell}(r)\Big) , \qquad \text{on $[r_0,\infty)$} \lv\\
H_{\ell}'(r_0) =0 \\
H_{\ell}(r_0) = 0
\end{cases}
\end{align}

Plugging in $r=r_0$ in the ODE and using the fact that $F_{\ell}(r_0) =1>0$, it follows that $\frac{(r-2m_0)^2}{r^2} H_{\ell}'(r)$ is increasing near $r_0$, implying that $H_{\ell}'$, and hence $H_{\ell}$, are positive near $r=r_0$. Using a standard bootstrap method (see Lemma 5.21 in \cite{ahmed}) and using the fact that $F_{\ell}(r) >0$ for $r\in [r_0, \infty)$, we have that $H_{\ell}(r) >0$ for $r\in [r_0,\infty)$. This implies that $F_{\ell+1}(r) = H_{\ell}(r) + F_{\ell}(r) >0$ for $r\in [r_0,\infty)$, and in particular, $\lim_{r\to \infty} r^{-1}F_{\ell+1} (r) = \infty$ as needed.

\end{proof}

We have then finally shown that $\Phi_{m\ell}^{\mathcal{G}}$ is an isomorphism from $H^2_{\delta}([r_0,\infty)) \cap C^2_{\delta}[r_0,\infty)$ to $L^2_{\delta-2}([r_0,\infty)) \cap C^0_{\delta-2}[r_0,\infty)$. This shows that there exists a unique solution $b_{m\ell}$ to the system in \eqref{sys2}.

\vv

We will now show that the function $b$ on $M$, with spherical harmonic coefficients $\sqrt{r(r-2m_0)} b_{m\ell}$ indeed lives in $\AHC{2}{k+2}{\delta+1}(M)$.

\vv
By the estimates in equations \eqref{estimate-f H N} and \eqref{estimate-f C N}, we estimate

\begin{align} 
&\normmmH{b''_{m\ell}}{0}{\delta-2}^2 + \left[ 1+\ell(\ell+1)\right] \normmmH{b'_{m\ell}(r)}{0}{\delta-1}^2 + \left[ 1+\ell(\ell+1)\right]^2 \normmmH{b_{m\ell}(r)}{0}{\delta}^2 \nonumber \\& \lesssim \normmmH{b'_{m\ell}(r)}{0}{\delta-2}^2 + \normmmH{b_{m\ell}(r)}{0}{\delta-1}^2\\  &\quad + \normmmH{\alpha'_{m\ell}}{0}{\delta-2}^2 + \normmmH{\alpha_{m\ell}}{0}{\delta-1}^2 + \normmmH{\zeta_{m\ell}}{0}{\delta-2}^2 + \left[ 1+\ell(\ell+1)\right]^{1/2}|\xi_{m\ell}|^2
\end{align}

\begin{align} 
&\normmmC{b''_{m\ell}}{0}{\delta-2}^2 + \left[ 1+\ell(\ell+1)\right] \normmmC{b'_{m\ell}(r)}{0}{\delta-1}^2 + \left[ 1+\ell(\ell+1)\right]^2 \normmmC{b_{m\ell}(r)}{0}{\delta}^2 \nonumber \\ &\lesssim  \normmmC{b'_{m\ell}(r)}{0}{\delta-2}^2 + \normmmC{b_{m\ell}(r)}{0}{\delta-1}^2\\ &\quad \normmmC{\alpha'_{m\ell}}{0}{\delta-2}^2 + \normmmC{\alpha_{m\ell}}{0}{\delta-1} + \normmmC{\zeta_{m\ell}}{0}{\delta-2} + \left[ 1+\ell(\ell+1)\right]|\xi_{m\ell}|^2
\end{align}

After absorbing the $\norm{b_{m\ell}}$, $\norm{b'_{m\ell}}$ terms to the left side, multiplying by $[1+\ell(\ell+1)]^k$ and summing over $\ell$ and $m$, we deduce using proposition \ref{sph-prop} that 
\begin{equation}
\norm{b}_{\AHC{2}{k+2}{\delta+1}}^2 \lesssim \norm{\alpha}_{\AHC{1}{k+1}{\delta}}^2 + \norm{\zeta}_{\AHC{1}{k}{\delta-1}}^2 + \norm{\Lambda}_{\Omega^k(S^2)}^2
\end{equation}
implying that $b \in \AHC{2}{k+2}{\delta+1}(M)$ as needed.

\section*{Step 2: Solving for $a'_{m\ell}(r_0)$, $a_{m\ell}(r_0)$, and $c_{m\ell}(r_0)$}

For $\ell = 0$, we compute $c_{00}(r_0)$ using the system in \eqref{sys1} to be 

\begin{equation}
c_{00}(r_0) = -\frac{1}{2(r_0-2m_0)} h_{00}
\end{equation}

satisfying the estimate 

\begin{equation}
|c_{00}(r_0)|^2 \lesssim |h_{00}|^2
\end{equation}

 \vv
 
 We now consider the case $\ell\geq 1$. The ODE in \eqref{sys3} evaluated at $r_0$ together with the two boundary conditions in \eqref{sys3} determine uniquely $a_{m\ell}'(r_0)$, $a_{m\ell}(r_0)$ and $c_{m\ell}(r_0)$ in terms of $\beta'_{m\ell}(r_0)$, $\lambda_{m\ell}$ and $h_{m\ell}$ to be

\begin{equation}
a_{m\ell}(r_0) = \frac{-h_{m\ell} + (r_0-2m_0) \beta'_{m\ell}(r_0) - (r_0-2m_0)\lambda_{m\ell}}{\underbrace{\ell(\ell+1) +2- \frac{6m_0}{r_0}}_{>0 \text{  for $\ell\geq 1, r_0>2m_0$}}}
\end{equation}

\begin{equation}
c_{m\ell}(r_0) = \frac{(r_0-3m_0)}{r_0(r_0-2m_0)}a_{m\ell}(r_0) + \frac{1}{2} (-\beta'_{m\ell}(r_0) + \lambda_{m\ell})
\end{equation}

\begin{equation}
a'_{m\ell}(r_0) = c_{m\ell}(r_0) + \beta'_{m\ell}(r_0)
\end{equation}

Note that we have used the fact that $\beta_{m\ell}(r_0) = 0$. We also have the following estimates:

\begin{align}
|a_{m\ell}(r_0)|^2 \lesssim \frac{1}{[1+\ell(\ell+1)]^2} (  |\lambda_{m\ell}|^2 + |\beta'_{m\ell}(r_0)|^2 + |h_{m\ell}|^2)
\end{align}

\begin{align} 
|c_{m\ell}(r_0)|^2  & \lesssim |a_{m\ell}(r_0)|^2 +  |\lambda_{m\ell}|^2 + |\beta'_{m\ell}(r_0)|^2\\
& \lesssim |\lambda_{m\ell}|^2 + |\beta'_{m\ell}(r_0)|^2 + \frac{1}{[1+\ell(\ell+1)]^2}|h_{m\ell}|^2
\end{align}

\begin{align}
|a'_{m\ell}(r_0)|^2 &\lesssim  |c_{m\ell}(r_0)|^2 + |\beta'_{m\ell}(r_0)|^2\\
&\lesssim |\lambda_{m\ell}|^2 + |\beta'_{m\ell}(r_0)|^2 + \frac{1}{[1+\ell(\ell+1)]^2}|h_{m\ell}|^2
\end{align}

\section*{Step 3: Solving for $c$}

In light of proposition \ref{ode-prop}, the ODE for $c_{00}$ in \eqref{sys1} and for $c_{m\ell}(r_0)$ in \eqref{sys3} is uniquely solvable in $H^2_{\delta}([r_0,\infty)) \cap C^2_{\delta}[r_0,\infty)$ with estimates: 

\begin{align} 
&\normmmH{c_{m\ell}''}{0}{\delta-2}^2 + \left[ 1+\ell(\ell+1)\right] \normmmH{c_{m\ell}'}{0}{\delta-1}^2 + \left[ 1+\ell(\ell+1)\right]^2 \normmmH{c_{m\ell}}{0}{\delta}^2 \\ &\lesssim  \left[ 1+\ell(\ell+1)\right]^{3/2}|c_{m\ell}(r_0)|^2 \\
&\lesssim \left[ 1+\ell(\ell+1)\right]^{3/2} \bigg( |\lambda_{m\ell}|^2 + |\beta'_{m\ell}(r_0)|^2 + \frac{1}{[1+\ell(\ell+1)]^2}|h_{m\ell}|^2 \bigg)
\end{align}
and 
\begin{align} 
&\normmmC{c_{m\ell}''}{0}{\delta-2}^2 + \left[ 1+\ell(\ell+1)\right] \normmmC{c_{m\ell}'}{0}{\delta-1}^2 + \left[ 1+\ell(\ell+1)\right]^2 \normmmC{c_{m\ell}}{0}{\delta}^2 \\  &\lesssim  \left[ 1+\ell(\ell+1)\right]^2|c_{m\ell}(r_0)|^2\\
&\lesssim   \left[ 1+\ell(\ell+1)\right]^2\bigg( |\lambda_{m\ell}|^2 + |\beta'_{m\ell}(r_0)|^2 + \frac{1}{[1+\ell(\ell+1)]^2}|h_{m\ell}|^2 \bigg)
\end{align}

Multiplying by $[1+\ell(\ell+1)]^{k-1}$ and summing over $\ell$ and $m$, we deduce using proposition \ref{sph-prop} that 
\begin{equation}
\norm{c}_{\AHC{2}{k+1}{\delta}}^2 \lesssim  \norm{\beta}_{\AHC{2}{k+2}{\delta+1}}^2+\norm{\Lambda}_{\Omega^k(S^2)}^2 + \norm{h}_{H^k(S^2)}^2
\end{equation}
implying that $c \in \AHC{2}{k+1}{\delta}(M)$ as needed.

\vv

\section*{Step 4: Solving for $a$}

For $\ell \geq 1$ and $-\ell\leq m \leq \ell$, define the function $\bar c_{m\ell}$ on $[r_0,\infty)$ by 

\begin{equation}
\bar c_{m\ell}(r) = \int_{r_0}^r c_{m\ell}(s) ds
\end{equation}

By integrating the ODE that $c_{m\ell}$ satisfies in \eqref{sys3}, we compute

\begin{equation}
\bar c_{m\ell}(r) = \frac{1}{\ell(\ell+1)} \Big( r(r-2m_0) c'_{m\ell}(r) - r_0(r_0-2m_0) c'_{m\ell}(r_0)\Big) 
\end{equation}

implying the following estimates:

\begin{align} 
& \normmmH{\bar c_{m\ell}}{0}{\delta+1}^2 \\
&\lesssim \frac{1}{\left[ 1+\ell(\ell+1)\right]^2} (\normmmH{ c_{m\ell}'}{0}{\delta-1}^2 + |c_{m\ell}'(r_0)|^2)\\
&\lesssim \frac{1}{\left[ 1+\ell(\ell+1)\right]^2} (\normmmH{ c_{m\ell}'}{0}{\delta-1}^2 + \ell(\ell+1)|c_{m\ell}(r_0)|^2)\\
\end{align}
and 
\begin{align} 
&\normmmC{\bar c_{m\ell}}{0}{\delta+1}^2 \\
&\lesssim \frac{1}{\left[ 1+\ell(\ell+1)\right]^2} (\normmmC{ c_{m\ell}'}{0}{\delta-1}^2 + |c_{m\ell}'(r_0)|^2)\\
&\lesssim \frac{1}{\left[ 1+\ell(\ell+1)\right]^2} (\normmmC{ c_{m\ell}'}{0}{\delta-1}^2 + \ell(\ell+1)|c_{m\ell}(r_0)|^2)\\
\end{align}

where we used the Rellich (Neumann-tangential) boundary inequality $|c'_{m\ell}(r_0)|\lesssim [\ell(\ell+1)]^{1/2} |c_{m\ell}(r_0)|$ (see also section 3, equation 3.26 and 3.27 in \cite{ahmed}). 
\vv

Now we solve for $a_{m\ell}$. The ODE for $a_{m\ell}$ in \eqref{sys3} is uniquely solvable in $H^2_{\delta+1}([r_0,\infty)) \cap C^2_{\delta+1}[r_0,\infty)$ with solution 

\begin{equation}
a_{m\ell}(r) = \bar c_{m\ell}(r) + \beta_{m\ell}(r) + a_{m\ell}(r_0)
\end{equation}

satisfying the estimates: 

\begin{align} 
&\normmmH{a_{m\ell}''}{0}{\delta-1}^2 + \left[ 1+\ell(\ell+1)\right] \normmmH{a_{m\ell}'}{0}{\delta}^2 + \left[ 1+\ell(\ell+1)\right]^2 \normmmH{a_{m\ell}}{0}{\delta+1}^2 \\ &\lesssim  \normmmH{c'_{m\ell}}{0}{\delta-1}^2 + \normmmH{\beta''_{m\ell}}{0}{\delta-1}^2 + [1+\ell(\ell+1)] (\normmmH{c_{m\ell}}{0}{\delta}^2 +  \normmmH{\beta'_{m\ell}}{0}{\delta}^2) \nonumber\\ &\qquad + [1+\ell(\ell+1)]^2( \normmmH{\bar c_{m\ell}}{0}{\delta+1}^2 + \normmmH{\beta_{m\ell}}{0}{\delta+1}^2) + [1+\ell(\ell+1)]^2 |a_{m\ell}(r_0)|^2   \\
&\lesssim \normmmH{\beta''_{m\ell}}{0}{\delta-1}^2 + [1+\ell(\ell+1)]  \normmmH{\beta'_{m\ell}}{0}{\delta}^2 + [1+\ell(\ell+1)]^2 \normmmH{\beta_{m\ell}}{0}{\delta+1}^2 \nonumber \\ 
&\qquad + [1+\ell(\ell+1)]^{1/2} (|\lambda_{m\ell}|^2 + |\beta'_{m\ell}(r_0)|^2) + |h_{m\ell}|^2 
\end{align}
and 
\begin{align} 
&\normmmC{a_{m\ell}''}{0}{\delta-1}^2 + \left[ 1+\ell(\ell+1)\right] \normmmC{a_{m\ell}'}{0}{\delta}^2 + \left[ 1+\ell(\ell+1)\right]^2 \normmmC{a_{m\ell}}{0}{\delta+1}^2 \\ &\lesssim  \normmmC{c'_{m\ell}}{0}{\delta-1}^2 + \normmmC{\beta''_{m\ell}}{0}{\delta-1}^2 + [1+\ell(\ell+1)] (\normmmC{c_{m\ell}}{0}{\delta}^2 +  \normmmC{\beta'_{m\ell}}{0}{\delta}^2) \nonumber\\ &\qquad + [1+\ell(\ell+1)]^2( \normmmC{\bar c_{m\ell}}{0}{\delta+1}^2 + \normmmC{\beta_{m\ell}}{0}{\delta+1}^2) + [1+\ell(\ell+1)]^2 |a_{m\ell}(r_0)|^2   \\
&\lesssim \normmmC{\beta''_{m\ell}}{0}{\delta-1}^2 + [1+\ell(\ell+1)]  \normmmC{\beta'_{m\ell}}{0}{\delta}^2 + [1+\ell(\ell+1)]^2 \normmmC{\beta_{m\ell}}{0}{\delta+1}^2 \nonumber \\ 
&\qquad + [1+\ell(\ell+1)] (|\lambda_{m\ell}|^2 + |\beta'_{m\ell}(r_0)|^2) + |h_{m\ell}|^2 
\end{align}

where we have repeatedly used the estimates for $c_{m\ell}$, $\bar c_{m\ell}$, $a_{m\ell}(r_0)$ that have already been established. 

\vv

Multiplying by $[1+\ell(\ell+1)]^k$ and summing over $\ell$ and $m$, we deduce using proposition \ref{sph-prop} that 
\begin{equation}
\norm{a}_{\AHC{2}{k+2}{\delta+1}}^2 \lesssim \norm{\beta}_{\AHC{2}{k+2}{\delta+1}}^2 + \norm{\Lambda}_{\Omega^k(S^2)}^2 + \norm{h}_{H^k(S^2)}^2
\end{equation}
implying that $a \in \AHC{2}{k+2}{\delta+1}(M)$ as needed.

\bibliographystyle{amsplain}
\bibliography{refs}

\section*{Data Availability Statement}
Data sharing is not applicable to this article as no datasets were generated or analyzed in this paper.

\section*{Declaration} 

\textbf{Conflict of Interest:} The authors have no Conflict of interest to declare that are relevant to the content of this
article.

\end{document}